\documentclass[journal]{IEEEtran}
\usepackage{amssymb}
\usepackage{amsmath}
\usepackage{amsthm}
\usepackage{mathrsfs}
\usepackage{graphicx}

\usepackage{subcaption}

\usepackage{subeqnarray}
\usepackage{cite}
\usepackage{color}
\usepackage{extarrows}
\usepackage{mathtools}
\usepackage{dsfont}
\usepackage{float}
\usepackage{enumerate}
\usepackage{marvosym}
\usepackage{upgreek}
\usepackage{url}
\usepackage{bm}

\usepackage{multirow}
\usepackage{multicol}
\usepackage{array}
\usepackage{arydshln}
\usepackage{booktabs}  % for table
\usepackage{makecell}

\usepackage[end]{algpseudocode}  % for algorithm
\usepackage{algorithmicx}
\usepackage{algorithm}

\DeclareMathOperator{\sgn}{sgn}

\def\bD{\ensuremath{{\bf D}}}
\def\bG{\ensuremath{{\bf G}}}

\def\bE{\ensuremath{{\bf E}}}

\def\bM{\ensuremath{{\bf M}}}

\def\bx{\ensuremath{{\bf x}}}

\def\bn{\ensuremath{{\bf n}}}

\def\by{\ensuremath{{\bf y}}}

\def\bK{\ensuremath{{\bf K}}}

\def\bD{\ensuremath{{\bf D}}}

\def\bu{\ensuremath{{\bf u}}}
\def\bk{\ensuremath{{\bf k}}}

\def\bx{\ensuremath{{\bf x}}}

\def\bE{\ensuremath{{\bf E}}}

\def\bD{\ensuremath{{\bf D}}}

\def\bw{\ensuremath{{\bf w}}}
\def\bx{\ensuremath{{\bf x}}}

\def\bE{\ensuremath{{\bf E}}}

\def\bK{\ensuremath{{\bf K}}}

\def\by{\ensuremath{{\bf y}}}

\def\bM{\ensuremath{{\bf M}}}

\newtheorem{theorem}{Theorem}
\newtheorem{lemma}{Lemma}

\newtheorem{assumption}{Assumption}
\begin{document}
\title{Deep, convergent, unrolled half-quadratic splitting for image deconvolution}
\author
{Yanan~Zhao,~%~\IEEEmembership{Student~Member,~IEEE,}
Yuelong~Li,~%~\IEEEmembership{Member,~IEEE,}
Haichuan~Zhang,~%~\IEEEmembership{Student~Member,~IEEE,}
Vishal~Monga,~%~\IEEEmembership{Senior~Member,~IEEE,}
and Yonina C. Eldar%~\IEEEmembership{Fellow,~IEEE}
%\thanks{{Manuscript received \today.}}
\thanks{Yanan Zhao is with the School of Electrical and Electronic Engineering, Nanyang Technological University, Singapore.}% <-this % stops a space
\thanks{Yuelong Li is with Amazon Lab 126, San Jose, CA, 94089 USA }
\thanks{Haichuan Zhang and Vishal Monga are with the Department of Electrical Engineering, The Pennsylvania State University, USA}
\thanks{Yonina C. Eldar is with the Department of Electrical Engineering, Faculty of Math and Computer Science, Weizmann Institute of Science, Israel.}
}
\markboth{submitted to IEEE TRANSACTIONS ON COMPUTATIONAL IMAGING}%
{Shell \MakeLowercase{\textit{et al.}}: Bare1 Demo of IEEEtran.cls for Journals}

\maketitle
\begin{abstract}
  In recent years, algorithm unrolling has emerged as a powerful technique for designing interpretable neural networks based on iterative algorithms.  Imaging inverse problems have particularly benefited from unrolling-based deep network design since many traditional model-based approaches rely on iterative optimization. Despite exciting progress, typical unrolling approaches heuristically design layer-specific convolution weights to improve performance. Crucially, convergence properties of the underlying iterative algorithm are lost once layer-specific parameters are learned from training data. We propose an unrolling technique that breaks the trade-off between retaining algorithm properties while simultaneously enhancing performance. We focus on image deblurring and unrolling the widely-applied Half-Quadratic Splitting (HQS) algorithm. We develop a new parametrization scheme which enforces layer-specific parameters to asymptotically approach certain fixed points. Through extensive experimental studies, we verify that our approach achieves competitive performance with state-of-the-art unrolled layer-specific learning and significantly improves over the traditional HQS algorithm. We further establish convergence of the proposed unrolled network as the number of layers approaches infinity, and characterize its convergence rate.  Our experimental verification involves simulations that validate the analytical results as well as comparison with state-of-the-art non-blind deblurring techniques on benchmark datasets. The merits of the proposed convergent unrolled network are established over competing alternatives, especially in the regime of limited training.
\end{abstract}

\section{Introduction}\label{sec:introduction}
    \par\IEEEPARstart{I}{mage} deconvolution refers to the process of recovering a sharp image from a recorded image corrupted by blur and noise. In the computational imaging literature, motion deblurring remains important since camera shakes are common in photography. The problem continues to attract attention because of the popularity of smartphone cameras, wherein effective hardware solutions such as professional camera stabilizers are difficult to deploy due to cost and space restrictions \cite{Wang2022robustloc,she2023robustmat,zhao2023convergent}. Therefore, deblurring algorithms are highly desirable.

Image deconvolution consists of blind and non-blind image deconvolution. The task of blind image deconvolution is to recover the sharp image given only the blurry image, without knowing the blur kernels~\cite{Perrone2016,Pan2018,Tao2018,ge2022blind}. Non-blind image deconvolution, where precise knowledge of the blur kernels are assumed known \emph{a priori}, continues to be an interesting topic despite decades of algorithmic developments.
In a typical computational imaging system, the performance of non-blind deconvolution
algorithms usually determines the quality of the reconstructed images. We concentrate on non-blind image deconvolution with the particular case of motion deblurring when conducting experimental studies. Nevertheless, our
formulation and analysis are versatile across blur kernels.

Early model-based works on image deblurring include Wiener
filters~\cite{wiener1949extrapolation} and Richardson-Lucy
iterations~\cite{Richardson1972}. More recent methods benefit from studies
about natural image statistics, which demonstrate that the image gradients often follow a sparse distribution. A popular line of research thus focuses on modeling the statistics of image gradients. One example is the Total-Variation (TV)
prior~\cite{rudin1992nonlinear}, which has the favorable property of being
convex and easy to optimize. However, in reality the gradient distribution
actually follows a hyper-Laplacian prior~\cite{levin2007image,Levin2011,ZhaKanSon:C23b,KanSonDinTay:C21}, which could be computationally demanding. Krinshnan {\it et
al.}~\cite{krishnan2009fast} proposed a fast approximate algorithm based on a
Look-Up-Table (LUT). However, since gradients are low-level image features,
these methods do not capture the high-level and non-local correlations in
natural images, which has been shown to be an important feature.
Danielyan {\it et al.}~\cite{danielyan2011bm3d} recast the Block
Matching 3-D (BM3D) algorithm as a frame operator, and integrate it into the analysis and synthesis deblurring framework, by imposing sparsity constraints over the transform coefficients. In doing so, they are effectively modeling the image statistics at a patch level rather than merely focusing on neighboring pixels. More recently, M{\"a}kinen {\it et al.} \cite{makinen2020collaborative} introduced a method for exact computation of the noise variance in a transform domain and embeds the new variance calculation into the BM3D algorithm, which yields significant improvements in BM3D denoising both visually and in Peak-Signal-to-Noise Ratio (PSNR).

All aforementioned techniques hinge on analytic or handcrafted prior models,
which may not faithfully fit real data distributions. Therefore, many works strive to learn prior models from real data. In~\cite{zoran2011learning},
Zoran {\it et al.} learn prior models from real image patches using a Gaussian Mixture Model (GMM), and integrate them into the well-known Half-Quadratic Splitting (HQS) framework dubbed Expected Patch Log Likelihood (EPLL) framework with GMM prior. As a follow-up, Sun {\it et al.} further extend this
prior into multiple scales, and develop specific priors by sampling patches
from similar images. While being more effective than conventional gradient
sparsity priors, these methods are typically slow. Rather
than learning the patch distribution, Schmidt and
Roth~\cite{schmidt2014shrinkage} learn the form of the shrinkage operator in HQS, by approximating it with a Gaussian Radial Basis Function (RBF), and learning the parameters through stage-by-stage training. Buzzard {\it et al.} \cite{buzzard2018plug} used the idea of Consensus Equilibrium (CE), which broadens the inclusion of a much wider variety of both forward components and prior components without the need for either to be expressed. This approach enables the use of trained DnCNNs, as proposed in \cite{zhang2017beyond}, at five different noise levels simultaneously in image denoising.

In recent years, deep neural networks have become the mainstream approach in learning-based methods. Xu {\it et al.}~\cite{LXu2014} learn a deconvolution Convolutional Neural Network (CNN) by decomposing the convolution weights into separable filters, which enable them to use large weights in the network. A follow-up work~\cite{ren2018deep} revised the weights decomposition scheme through a low-rank approximation. Another line of research applies pre-deconvolution techniques. Rather than feeding the original blurry images into the network directly, one can obtain a deconvolved image by deblurring it with a simple deconvolution algorithm such as a Wiener filter, and then input the deconvolved image into a network to perform artifact removal. Schuler {\it et al.}~\cite{schuler2013machine} employ a Multi-Layer Perceptron (MLP) to denoise the deconvolved patches. Zhang {\it et al.}~\cite{zhang2017} train a set of fast and effective CNN denoisers and integrate them into a model-based optimization method of HQS to solve the image deblurring problem.
Zhang {\it et
al.}~\cite{zhang2017learning} stack a series of fully convolutional networks
to recover the sharp edges from the image gradients. Motivated by deep residual
learning in the image super-resolution literature~\cite{kim2016accurate}, Wang
and Tao~\cite{wang2018training} employ a very deep CNN to remove the artifacts of the deconvolved image, which are modeled as the residual signal. Son and Lee~\cite{son2017fast} further adopt the residual connections from ResNet~\cite{he2016deep} into their network. Dong {\it et al.} \cite{dong2021learning} proposed a spatially-variant MAP model (SVMAP) by embedding deep neural networks within the constraints of the maximum a-posteriori (MAP) framework. The work in \cite{dong2021dwdn} performs deconvolution in a feature space instead of standard image space, by combining the classic Wiener deconvolution with deep learning features. Fang {\it et al.} \cite{fang2022robust} propose a kernel error term to rectify the given kernel at the time of performing the deconvolution and a residual error term is to deal with the outliers caused by noise or saturation.  Zhang {\it et al.} \cite{zhang2023infwide} proposes a two-branch architecture that can perform high-quality night photograph deblurring with sophisticated noise and saturation regions. Aggarwal {\it et al.} \cite{aggarwal2018modl} introduce an image reconstruction framework with a CNN-based regularization prior alongside conjugate gradient (CG) optimization.

\par While deep learning is empirically successful with sufficient training data, a major drawback is that it usually sacrifices interpretability, as the deep networks are constructed by stacking typical regression layers rather than derived from physical mechanisms. Gregor {\it et al.}~\cite{Gregor2010,monga2021algorithm} develop a sparse coding technique which could be regarded as a bridge between traditional iterative algorithms and modern deep neural networks, and hence offers promise in filling the interpretability gap. Essentially, for a particular iterative algorithm each iteration step could be mapped into one network layer and executing the algorithm a finite number of times means stacking such layers together. These concatenated layers form a deep network, with algorithm parameters mapped into network parameters. The parameters can then be learned from real datasets via end-to-end training. In recent years, there is a growing trend of research along this route in various imaging fields.

Prior unrolling works on blind deblurring~\cite{Li2019,Li2020,nir2022model,huang2022unrolled,richmond2022nonuniform} have
already demonstrated the merits of this technique. For (photon limited) non-blind deblurring, Sanghvi {\it et al.} \cite{sanghvi2022photon} unroll a Plug-and-Play algorithm for a fixed number of iterations while specifically addressing 
 the scenario of  Poisson noise. However, to ensure good
performance in practice, the network layers deviate from the original
iterations in that layer-specific parameters are adopted, potentially jeopardizing convergence. Similar issues also persist in other
unrolling works~\cite{hershey2014deep}, where custom layers are
incorporated, and convergence guarantees for unrolled networks have become a common open question.

From an algorithmic perspective, HQS is a generic and widely applied technique
for solving prior-based non-blind image deblurring. The convergence of
conventional HQS is established in~\cite{wang2008new}, where Wang {\it et al.} derive strong convergence properties for the algorithm including finite
convergence for certain variables. In addition, they quantitively characterize
the convergence rate, which turns out to be at least $q$-linear~\cite{jay2001note} when fixing the
penalty parameters. In practice, the convergence speed could be further
accelerated when combined with a continuation scheme. Nevertheless, their
analysis is restricted to the scenarios where the parameters are fixed per
iterations, and hence cannot be directly applied to many unrolling approaches
where the parameters are altered.

\textbf{Motivations and Contributions:}
In existing unrolled deep networks for blind image deblurring~\cite{Li2020,huang2022unrolled,richmond2022nonuniform}, 
significant performance gains (in the sense of reconstructed image quality) have been demonstrated by layer-specific learning including the case of unrolled HQS \cite{Li2020,Mouna2023}. The introduction of layer-specific learned parameters such as independent filters per layer also has a significant drawback in that it undermines the convergence guarantees associated with traditional HQS, whose parameters are recurrent across
iterations.

In this paper, we propose a custom unrolled HQS network that preserves the analytical merits of HQS while
simultaneously retaining the practical performance benefits of unrolling. Different from prior independent layer-specific learning, we develop a \textit{structured} per-layer parameterization which we prove effectively leads to convergence. The main contributions of
this paper are as follows:

\begin{itemize}
  \item [$\bullet$] \textbf{Deep Convergent Unrolling for Non-blind deblurring (DECUN)}: We propose a deep interpretable neural network called DECUN by unrolling the widely-applied HQS algorithm. To ensure convergence and network modeling power/performance, we develop a specific structured parametrization scheme.
  \item [$\bullet$] \textbf{Analysis of DECUN Properties:} We prove the convergence of DECUN under the aforementioned parametrization. Furthermore, we quantitatively characterize its convergence rate both in a general form and for certain special cases. We verify our theoretical findings via simulation studies.
  \item[$\bullet$] \textbf{Performance and Computational benefits}: By experimental comparison with state-of-the-art deblurring techniques, DECUN achieves a favorable gain of PSNR by about 1 dB and an SSIM of 0.1 over both traditional iterative algorithms and modern deep neural networks, while preserving convergence and interpretability.
\end{itemize}

We note that on the theoretical front there are works concentrating on analyzing the convergence of Learned Iterative Shrinkage and Thresholding Algorithm (LISTA). In particular, Chen {\it et al.}~\cite{chen2018theoretical} derive necessary conditions for the network parameters to obey in order to ensure convergence; Liu {\it et al.}~\cite{liu2019alista} further analytically specify a formula for the network weights. Our work is of the same spirit in that we also develop conditions for
parameters and analyze convergence; however, we focus on a different underlying iterative algorithm (HQS vs. LISTA) and carry out specific realistic experimental studies for non-blind image deblurring.

The rest of the article is organized as follows. In Section~\ref{sec:hqs_unrolling} we introduce both traditional and unrolled HQS. In Section~\ref{sec:hqs_analysis}, we discuss our new parameterization scheme, and provide rigorous convergence analysis and characterize the convergence rate of our proposed unrolled network, dubbed DECUN;\@ in Section~\ref{sec:experiemnts} we verify our theory through simulation studies and demonstrate the practical advantages of DECUN via extensive experimental studies including enhanced generalizability. Section~\ref{sec:conclusions} concludes the paper.
\section{Unrolled HQS for Image Deblurring}\label{sec:hqs_unrolling}
% \section{Unrolled Half-Quadratic Splitting for Image Deblurring}\label{sec:hqs_unrolling}
%In this section we discuss construction of the unrolled network. Such networks are typically backed by an underlying iterative algorithm. To this end, we first formulate the non-blind deconvolution problem and the iterative algorithm to solving it.

\subsection{Formulation of the Deblurring Problem}\label{ssec:deconv_formulation}
Assuming the camera motion is purely translational and the scene is planar (ignoring the depth variations), image blurring can be modeled as a discrete convolution process~\cite{kundur1996blind}:
 \vspace{-1mm}
\begin{equation}
  \by = \bk\ast \bu + \bn,\label{eq:conv_model}
  \vspace{-1mm}
\end{equation}
where $\by\in\mathbb{R}^m$ denotes the observations of the signal $\bu\in\mathbb{R}^n$, received from a system with impulse response $\bk\in\mathbb{R}^k$ and $\bn\in\mathbb{R}^m$ is an additive random noise process which is usually modeled as independent identically distributed (i.i.d.) Gaussian. By re-writing convolution as a matrix multiplication, we can cast~\eqref{eq:conv_model} into the following formula:
 \vspace{-1mm}
\begin{equation}
  \by = \bK\bu+\bn,\label{eq:linear_model}
  \vspace{-1mm}
\end{equation}
where $\bK\in\mathbb{R}^{m\times n}$ is the Toeplitz matrix corresponding to the convolution kernel $\bk$. In \emph{non-blind deconvolution}, it is assumed that $\bK$ is pre-determined and the goal is to solve for $\bu$.

In reality, $\bk$ is usually considered as a low-pass filter since it essentially performs temporal averaging of displaced versions of the original signal $\bu$. In this regard, $\bK$ is an ill-posed linear operator and the non-blind deconvolution problem is typically an \emph{ill-posed} linear inverse problem.
In order to recover $\bu$ faithfully, we pose additional assumptions on the structure of $\bu$. These assumptions can be expressed as statistical priors, or deterministic regularizers. In particular, a popular prior model leverages natural image statistics, which enforces gradient sparsity of the latent image $\bu$. This can be achieved by regularizing the $\ell^1$-norm of image gradients, which is well-known to be a convex function. This technique, commonly called \emph{total variation minimization} \cite{Goldstein09}, reduces to solving the following convex optimization problem:
\vspace{-1mm}
\begin{equation}
    \min_{\bu\in\mathbb{R}^n} \frac{\mu}{2}\|\by-\bK\bu\|_{2}^{2}+\|\bD_x\bu\|_1+\|\bD_y\bu\|_1\label{eq:tv_deconv},
    \vspace{-1mm}
\end{equation}
where $\bD_x,\bD_y\in\mathbb{R}^{n\times n}$ are the gradient operators which capture horizontal and vertical image gradients, respectively, and $\mu>0$ is a regularization parameter that controls the strength of enforcing gradient sparsity. In practice, image gradients can be computed by linear filtering and $\bD_x,\bD_y$ are Toeplitz operators. For expositional ease, we identify the operators $\bD_x$ and $\bD_y$ with their underlying filters henceforth.
\vspace{-4mm}
\subsection{Solving Non-blind Deconvolution Using HQS}\label{ssec:half_quadratic_splitting}
%\subsection{Solving Non-blind Deconvolution Using the Half-Quadratic Splitting Algorithm}\label{ssec:half_quadratic_splitting}
A widely employed algorithm for solving~\eqref{eq:tv_deconv} is the % \emph{Half-Quadratic Splitting  (HQS)}
\emph{HQS} algorithm~\cite{wang2008new}. By introducing auxiliary variables $\bw_x,\bw_y\in\mathbb{R}^n$ as surrogates of image gradients, \eqref{eq:tv_deconv} can be cast into the following constrained minimization problem:
\vspace{-1mm}
\begin{equation}\label{eq:var_split}
    \begin{split}
        \min_{\bu,\bw_x,\bw_y\in\mathbb{R}^n}&\quad\frac{\mu}{2}\|\by-\bK\bu\|_{2}^{2}+\|\bw_x\|_{1}+\|\bw_y\|_1,\\
        \text{subject to }&\quad\bD_x\bu=\bw_x,\quad\bD_y\bu=\bw_y.
    \end{split}
\end{equation}
We can relax the equality constraints in~\eqref{eq:var_split} to yield:
\vspace{-1mm}
\begin{equation}\label{eq:hqs_form}
    \begin{split}
        &\min_{\bu,\bw_x,\bw_y\in\mathbb{R}^n}\quad\frac{\mu}{2}\|\by-\bK\bu\|_{2}^{2}+\|\bw_x\|_{1}+\|\bw_y\|_1\\
        &\qquad\qquad+\frac{\beta}{2}\left(\|\bD_x\bu-\bw_x\|_{2}^{2}+\|\bD_y\bu-\bw_y\|_2^2\right),
        \vspace{-1mm}
    \end{split}
\end{equation}
where $\beta>0$ is a parameter controlling the relaxation strength. The HQS algorithm proceeds by minimizing over $\bu$ and $\bw_x,\bw_y$ in an alternating fashion. A desirable property of~\eqref{eq:hqs_form} is that each sub-problem (over $\bu$ or over $\bw$) admits closed-form solutions, enabling to express the iteration steps analytically for $l=1,2,\dots$:
\vspace{-2mm}
\begin{equation}\label{eq:hqs_iter}
    \begin{split}
        %\text{ for }&l=1,2,\dots\\
        \bM&\gets\sum_{i\in{x,y}} {\bD_i}^\mathrm{T} \bD_i + \frac{\mu}{\beta}\bK^\mathrm{T} \bK,\\
        \bu^l&\gets\bM^{-1}\left(\sum_{i\in{x,y}} {\bD_i}^\mathrm{T} \bw^l_i+\frac{\mu}{\beta}\bK^\mathrm{T} \by\right)\\
        \bw_x^{l+1}&\gets s_{\frac{1}{\beta}}\left(\bD_x\bu^l\right),\quad\bw_y^{l+1}\gets s_{\frac{1}{\beta}}\left(\bD_y\bu^l\right),
\vspace{-2mm}
    \end{split}
\end{equation}
where $s_\lambda$ is the soft-thresholding operator defined elementwise as $s_\lambda(x)=\sgn(x)\cdot\max\{|x|-\lambda,0\}$. Since $\bD_x,\bD_y,\bK$ are all Toeplitz operators, matrix inversion can be implemented efficiently via the Discrete Fourier Transform (DFT).

Convergence of HQS has been established in~\cite{wang2008new} under mild technical assumptions. Furthermore, it has been shown that the convergence rate is at least $q$-linear, namely:
\vspace{-1mm}
    \[
        \|\bw^{l+1}-\bw^\ast\|\leq\lambda\|\bw^l-\bw^\ast\|,
        \vspace{-1mm}
    \]
where \footnote{In our analysis $\lambda$ plays the role of $q$.}$\lambda<1$ depends on $\bD_x,\bD_y$ and $\bK$. The analysis in \cite{wang2008new} hinges on the fact that $\bD_x,\bD_y$ and $\beta$ are \textit{fixed} parameters across iterations, and does not apply to unrolling since learned parameters change per layer/iteration.
\vspace{-4mm}
\subsection{Deep Unrolled Half-Quadratic Splitting}\label{ssec:deep_hqs_unroll}
In conventional HQS iterations~\eqref{eq:hqs_iter}, the operators
$\bD_x,\bD_y$, or their underlying gradient filters, are determined
analytically. The parameters $\beta,\mu$ are also chosen as fixed values. In practice, better filters and parameters can be learned from real datasets through end-to-end training to improve performance. Indeed, as discussed
in~\cite{Gregor2010}, for a particular iterative algorithm, each iteration step
can be considered as one layer of a network, and executing the iterative steps
corresponds to concatenating these layers together to form a deep network. For
instance, the iteration steps in~\eqref{eq:hqs_iter} comprise a series of
linear mappings, followed by non-linear soft-thresholding operations, which
together form a single network layer. In this way, the HQS algorithm is unrolled into a deep network, which we train end-to-end to optimize the
filters and parameters.

%In addition, the filters $\bD_x,\bD_y$ and parameters $\mu,\beta$ are fixed per iterations, which we believe will hamper the modeling power of the unrolled network and hinder its performance. 
Therefore, we further generalize HQS by adopting different parameters $\bD_x^l,\bD_y^l,\beta^l$ per layer. In addition, for each layer we use $C$ filters instead of two, i.e., we embed ${\{\bD_i^l\}}_{i=1}^C$ into each layer. With these generalizations, each network layer now comprises the following iterations:
\vspace{-2mm}
\begin{equation}\label{eq:unroll_hqs_iter}
    \begin{split}
        \bM^l&\gets\sum_{i=1}^C {\bD^{l}_i}^\mathrm{T} \bD^{l}_i + \frac{\mu}{\beta^{l}}\bK^\mathrm{T} \bK,\\
        \bu^l&\gets{\bM^l}^{-1}
        \left(\sum_{i=1}^C {\bD_i^{l}}^\mathrm{T} \bw_i^{l}+\frac{\mu}{{\beta}^{l}}\bK^\mathrm{T} \by\right),\\
        \bw_i^{l+1}&\gets s_{\frac{1}{\beta^l}}\left(\bD_i^l \bu^l\right),\quad i=1,2,\dots,C.
    \end{split}
    \vspace{-2mm}
\end{equation}
To simplify the notations, let $\bw^l = {\left[{\bw_1^l}^T,{\bw_2^l}^T,\dots,{\bw_{C}^l}^T\right]}^T$, $\bD^l={\left[{\bD^l_1}^T,{\bD^l_2}^T,\dots,{\bD^l_C}^T\right]}^T$ and $s^{l}(\cdot)=s_{\frac{1}{\beta^l}}(\cdot)$. Then iterations~\eqref{eq:unroll_hqs_iter} can be re-written as:
\begin{align}
    \bM^l&\gets{\bD^{l}}^\mathrm{T} \bD^{l} + \frac{\mu}{\beta^{l}}\bK^\mathrm{T}\bK,\label{eq:para_M_l}\\
    \bu^{l}&\gets{\bM^l}^{-1}\left({\bD^{l}}^\mathrm{T} \bw^{l}+\frac{\mu}{{\beta}^{l}}\bK^\mathrm{T} \by\right),\label{eq:u_update}\\
    \bw^{l+1}&\gets s^{l}\left(\bD^l \bu^l\right).\label{eq:w_update}
    \vspace{-1mm}
\end{align}
Furthermore, define
\begin{equation}
  h^{l}(\bw^{l})= \bD^{l}\bu^{l}=\bD^{l}{\bM^l}^{-1}
           \left({\bD^{l}}^\mathrm{T} \bw^{l}+\frac{\mu}{{\beta}^{l}}\bK^\mathrm{T} \by\right).\label{eq:function_h_l}
\end{equation}
We may express the iterations succinctly as
\begin{align}
 \bw^{l+1}= s^{l}\left(h^{l}\left(\bw^{l}\right)\right). \label{w_l1}
\end{align}

With the mappings $s^{l}$, $h^{l}$ defined, the unrolled network can be visually depicted
as a diagram in Fig.~\ref{fig:network}. Compared to traditional HQS, unrolled HQS performs much better in practice thanks to its enriched set of parameters~\cite{Li2020}; furthermore, it typically executes merely tens of layers in practice, as opposed to hundreds or thousands of iterations, and is usually much more computationally efficient than HQS.
However, since $\bD^l$ and $\beta^l$ are layer-specific parameters learned with training data, the convergence guarantee for HQS no longer holds, as we will demonstrate empirically in Section~\ref{ssec:exp_simulation};
additionally, interpretability might be undermined because the unrolled
network, when having infinitely many layers, may no longer converge to a
particular fixed point. Nevertheless, as we will formally prove in Section~\ref{sec:hqs_analysis}, under certain conditions, if both the filters and the parameters are asymptotically fixed, i.e., both converge to certain fixed values, then the convergence guarantee could be preserved.

\begin{figure*}
    \centering
    \includegraphics[width=0.8\textwidth]{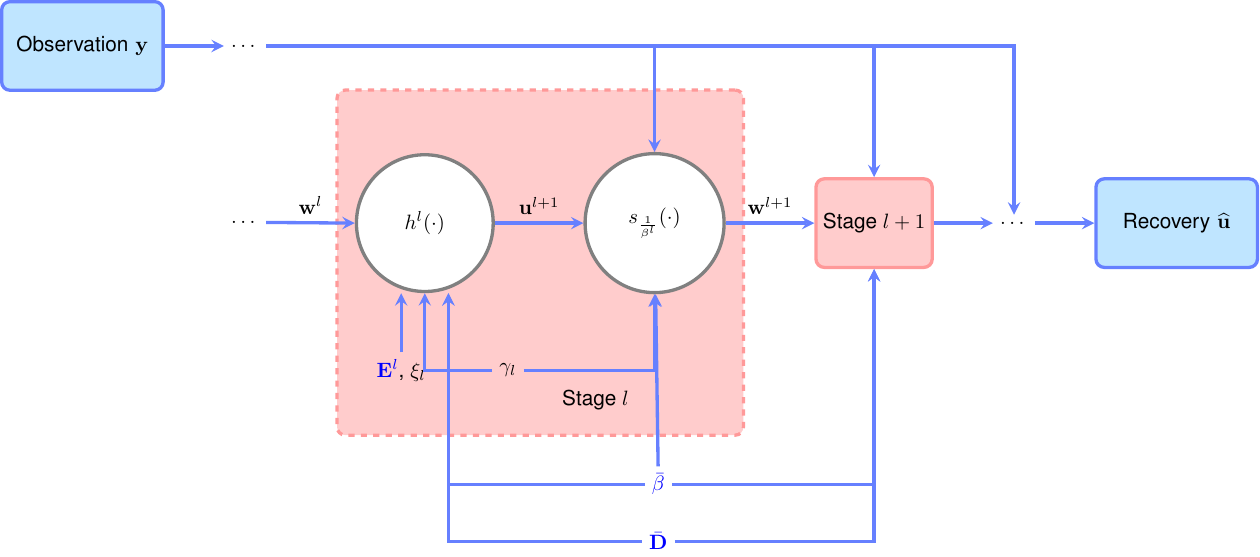}
    \caption{Diagram of the unrolled HQS network. Trainable parameters are colored in blue. The linear operator $h^l$ is defined in~\eqref{eq:function_h_l}, whereas $s^{l}$ is the soft-thresholding operator defined in~\eqref{S_operator}.}\label{fig:network}
\end{figure*}

\section{Deep Convergent, Unrolled HQS}\label{sec:hqs_analysis}
In this section, we first propose a deep interpretable convergent neural network architecture dubbed Deep Convergent Unrolling for Non-blind Deblurring (DECUN) by developing a new parametrization scheme in HQS.\@ Then, we provide a convergence proof and determine the convergence rate.

\subsection{Convergence Analysis of DECUN}
% \subsection{Convergence Analysis of Deep Convergent Unrolled for Non-blind Deblurring}
In~\cite{wang2008new}, convergence is established for fixed $\bD,\mu$ and $\beta$, where the mappings $s$ and $h$ are fixed per iteration. However, in unrolled HQS, the fixed point of $s^{l}\circ h^l$ is not defined. We hypothesize that, if ${\{\bD^l\}}_l$ and ${\{\beta^l\}}_l$ are convergent sequences, i.e., $\lim_{l\rightarrow\infty}\bD^l=\bar{\bD}$ and $\lim_{l\rightarrow\infty}\beta^l=\bar{\beta}$, then ${\{\bw^l\}}_l$ converges to a fixed point of $s^{\ast}\circ h^\ast$, with $s^{\ast},h^\ast$ being the asymptotic mappings of $s^{l}$ and $h^l$, respectively.

In particular, we propose the following parameterization:
\begin{equation}
    \bD^l=\bar{\bD}+\xi_l\bE^l,\quad\beta^l=\bar{\beta}+\gamma_l,
  \label{eq:parameterization}
\end{equation}
where $\xi_l$ and $\gamma_l$ are both real vanishing sequences with
$\lim_{l\rightarrow\infty}\xi_l=0$ and $\lim_{l\rightarrow\infty}\gamma_l=0$,
and $\bar{\bD},\bE^l$ are both Toeplitz matrices representing discrete convolutions, and $\bE^l$ is norm-bounded. Since $\bE^l$'s are layer-specific parameters, the number of parameters for $\bD^l$ does not decrease due to the parameterization scheme shown in \eqref{eq:parameterization}, which enables the unrolled network to still retain significant modeling power. In practice,
$\bar{\bD}$, $\bar {\beta}$ and $\bE^l$'s are learned from real-world training data whereas $\xi_l$ and
$\gamma_l$ are chosen real sequences. Next, we formally prove that,
this approach leads to convergence and characterizes the convergence rate.
\vspace{-2mm}
\begin{assumption}\label{assump}
For any iteration $l$, $\mathcal{N}(\bK)\cap \mathcal{N}(\bD^{l})=\{\mathbf{0}\}$, where $\mathcal{N}(\cdot)$ represents the null space of a matrix.
\end{assumption}
\vspace{-2mm}
This assumption ensures the nonsingularity of $\bM^l$ defined in \eqref{eq:para_M_l}, so that $\mathbf{G}^{l}=\bD^{l}{(\bM^{l})}^{-1}{(\bD^{l})}^\mathrm{T}$ could be well defined and used in Theorem \ref{thm:convergence_rate}.
\vspace{-2mm}
\begin{theorem}\label{convergence_l}
  (Convergence result) For a fixed $\mu > 0 $, under assumption \ref{assump} and the parameterization given in~\eqref{eq:parameterization}, suppose $\xi_{l}\in\mathbb{R}$, $\gamma_{l}\in\mathbb{R}$ are absolutely summable, meaning $\sum_l|\xi_l|$ and $\sum_l|\gamma_l|$ both converge. Furthermore, suppose that ${\{\bE^l\}}_l$ forms a bounded sequence, i.e., $\|\bE^l\|\leq M$ for some $M > 0$. Then the sequence $\{(\bw^{l},\bu^{l})\}$ generated by executing DECUN defined in \eqref{eq:u_update} and \eqref{eq:w_update} from any starting point $\{(\bw^{0},\bu^{0})\}$ converges to $\{(\bw^{\ast},\bu^{\ast})\}$ as $l\rightarrow\infty$.
\end{theorem}
\vspace{-3mm}
\begin{proof}
    See Appendix~\ref{sec:proof_thm1}.
\end{proof}
\vspace{-2mm}
Essentially, Theorem~\ref{convergence_l} dictates that once we execute DECUN for infinitely many layers, its output converges to a certain fixed point. In other words, DECUN corresponds to a convergent iterative algorithm, which dramatically enhances its interpretability. We further quantify the convergence rate of DECUN in the theorem below:
\vspace{-2mm}
\begin{theorem}\label{thm:convergence_rate}
  (Convergence rate) Under assumption \ref{assump}, let $\mathbf{G}^{l}=\bD^{l}{(\bM^{l})}^{-1}{(\bD^{l})}^\mathrm{T}$ where $\bM^{l}$ is defined in~\eqref{eq:para_M_l} and suppose that $\bD^l$ and $\beta^l$ follow the parametrization in~\eqref{eq:parameterization} under the conditions of Theorem~\ref{convergence_l}. Then the sequence $\{(\bw^{l},\bu^{l})\}$  generated by~\eqref{eq:u_update} and~\eqref{eq:w_update} converges to $\{(\bw^{*},\bu^{*})\}$ with convergence rate satisfying
\begin{equation}
\begin{split}
\label{convergence_rate_w}
\vspace{-1mm}
&\|\mathbf{w}^{l+1}-\mathbf{w}^{*}\|\\
&\quad\leq(\lambda_{\max})^{l+1}\|\mathbf{w}^{0}-\mathbf{w}^{*}\|+\mathcal{F}^{-1}\left(\frac{B(w)}{1-\lambda_{\max}e^{-jw}}\right)
\vspace{-1mm}
\end{split}
\end{equation}
where $\lambda_{\max}=\max\left\{\lambda_1,\lambda_2,\dots,\lambda_{l_0},\sqrt{\frac{1}{2}\left[1 + \rho({\mathbf{G}}^\ast)\right]}\right\}$ with $\lambda_{l}=\sqrt{\rho((\mathbf{G}^{l})^{2}))}$, and $\rho({\mathbf{G}}^\ast)$ and $\rho((\mathbf{G}^{l})^{2}))$ represent the spectral radii of matrix ${\mathbf{G}}^\ast$ and $(\mathbf{G}^{l})^{2}$ respectively. Here $B(w)$ is the discrete-time Fourier transform of
\vspace{-1mm}
\begin{equation}
    b_{l} = 3|\xi_{l}|\|\bu^{\ast}\|+\frac{2|\gamma_l|}{{(\bar{\beta})}^2} 
    \vspace{-1mm}
\end{equation}
and
\begin{equation}
\label{u_J}
\vspace{-1mm}
\bu^{\ast}=\left({{(\bar{\bD})}^\mathrm{T}}\bar{\bD} + \frac{\mu}{\bar{\beta}}\bK^\mathrm{T} \bK\right)^{-1}\!\left({(\bar{\bD})}^\mathrm{T}\bw^{\ast}+\frac{\mu}{\bar{\beta}}\bK^\mathrm{T} \by\right)
\vspace{-1mm}
\end{equation}
where $\bw^{\ast}$ is the fixed point of $s^{\ast}\circ h^\ast$: $\bw^{\ast}=s^{\ast}(h^\ast(\bw^{\ast}))$, which could be achieved from \eqref{eq:u_update} and \eqref{eq:w_update} as $l\rightarrow\infty$.
\begin{proof}
    See Appendix~\ref{sec:proof_thm2}.
\end{proof}
\end{theorem}
Theorem~\ref{thm:convergence_rate} provides a generic formula for the convergence rate of DECUN; however, it may be hard to pursue unless $B(w)$ is of a tractable form. In the following, we will derive simplified formulas of~\eqref{convergence_rate_w} under two common circumstances. Indeed, as can be seen from \eqref{convergence_rate_w}, the convergence rate depends on two parts. One depends on the spectral radii of $(\mathbf{G}^{l})^{2}$. It is easy to see that $\rho((\mathbf{G}^{l})^{2})\leq\rho((\mathbf{G}^{l}))< 1$ and $(\lambda_{\max})^{l+1}<<1$, since $\bG^l$ is symmetric, postive definite. Another depends on the convergence rate of both $\xi_l$ and $\gamma_l$. When $\xi_l=\xi^{l}$ with $0<\xi<1$ and $\gamma_l=\gamma^{l}$ with $0<\gamma<1$, then $B(w)$ in~\eqref{convergence_rate_w} yields
\begin{equation}
\begin{split}
    B(w)&=\frac{3\|\bu^{\ast}\|}{1-\xi e^{-jw}}+\frac{2}{{(\bar{\beta})}^2}\frac{1}{1-\gamma e^{-jw}}
\end{split}
\end{equation}
so that~\eqref{convergence_rate_w} yields
\begin{equation}
\begin{split}
\label{convergence_rate_w_1}
\vspace{-1mm}
&\|\mathbf{w}^{l+1}-\mathbf{w}^{*}\|\leq(\lambda_{\max})^{l+1}\|\mathbf{w}^{0}-\mathbf{w}^{*}\|\\
&\qquad+\frac{3\|\bu^{\ast}\|}{\xi-\lambda_{\max}}\left(\xi^{l+1}-(\lambda_{\max})^{l+1}\right)\\
&\qquad\qquad+\frac{2}{(\bar{\beta})^{2}(\gamma-\lambda_{\max})}(\gamma^{l+1}-(\lambda_{\max})^{l+1}).
\end{split}
\vspace{-1mm}
\end{equation}
If $\xi_{l}=\left(\frac{1}{l+1}\right)^{p}$, $\gamma_{l}=\left(\frac{1}{l+1}\right)^{p}$ with $p\geq 1$, then \eqref{convergence_rate_w} yields
\begin{equation}\label{convergence_rate_p}
\begin{split}
\vspace{-1mm}
&\|\mathbf{w}^{l+1}-\mathbf{w}^{*}\|\leq(\lambda_{\max})^{l+1}\|\mathbf{w}^{0}-\mathbf{w}^{*}\|\\
&\qquad + \left(3\|\bu^{*}\|+\frac{2}{(\bar{\beta})^{2}}\right)\sum_{i=0}^{l}(\lambda_{\max})^{l-i}\left(\frac{1}{i+1}\right)^{p}.
\vspace{-1mm}
\end{split}
\end{equation}

From~\eqref{convergence_rate_w_1},~\eqref{convergence_rate_p} we observe that the overall convergence rate depends on the decaying rate of $\xi_l$ and $\gamma_l$, and becomes faster when both series vanish quicker; in particular, when they are both geometric series then the convergence rate is linear. It is worth noting that in practice faster decay does not necessarily yield better performance because this merely implies that the speed of converging to a fixed point is high, which does not provide a guarantee that the fixed point is superior. Improved model performance can be obtained by better-chosen parameters which are learned from real datasets through end-to-end training in practice.

In~\cite{wang2008new}, Wang {\it et al.} conduct a similar analysis on the HQS algorithm (Theorem 3.4 and Theorem 3.5). In comparison, unrolled HQS generalizes it by altering the parameters across layers, and falls back to HQS when $\xi_l=0$ and $\gamma_l=0, \forall l$. In addition, when unrolled HQS converges, the convergent point is the same as HQS when using $\bar{\bD}$
and $\bar{\beta}$ as its parameters. In this regard, Theorem~\ref{convergence_l} could be considered as a generalization of Theorem 3.4 in~\cite{wang2008new} and covers a much wider variety of scenarios.
Compared to Theorem 3.5 in~\cite{wang2008new},
Theorem~\ref{thm:convergence_rate} provides a more sophisticated formula
for convergence rate since it deals with more complex cases. The convergence rate of DECUN expectedly is slower than that of traditional HQS, but in practice, it provides significant performance gains even with a small number of layers, as we will demonstrate in Section~\ref{sec:experiemnts}. Our simulation
studies in Section~\ref{ssec:exp_simulation} shed some light on how it simplifies
under common choices of $\xi_l$ and $\gamma_l$.
\vspace{-2mm}
%\textcolor{blue}{\subsection{Complexity Analysis}
% In this section, we discuss the complexity of our proposed DECUN. XXX}
\subsection{Connections to Related Analytical Work in Unrolled Deep Networks}
We compare our method against related works that also conduct theoretical analysis along with developing unrolled network architectures, such as~\cite{chen2018theoretical, liu2019alista, Chun2023}.
In contrasting our asymptotic convergence analyses against that of \cite{chen2018theoretical}, it introduces a partial weight coupling structure as a necessary condition for the convergence guarantee of LISTA, in addition to the sparsity assumption of the underlying signal to be recovered. This structure, primarily aimed at facilitating convergence analysis, might limit the model's capacity by constraining the number of parameters. In comparison, our proposed DECUN's parameterization scheme strikes a balance between analytical rigor and practical merits. It allows the network to maintain its capacity and, consequently, its practical performance, without compromising on the convergence guarantee. Furthermore, while LISTA primarily seeks to enhance the efficiency of sparse coding algorithms and generally does not surpass its analytical counterparts in recovering higher-quality sparse solutions, our experimental studies demonstrate DECUN's effectiveness in recovering latent images of superior quality.

Regarding the work of \cite{Chun2023}, while it also integrates neural networks (INN) into iterative reconstruction algorithms, our work differs significantly in the following aspects: 1) Chun \emph{et al} \cite{Chun2023} construct Momentum-Net by applying the Block Proximal Extrapolated Gradient method using a Majorizer (BPEG-M) framework to solve an iterative model-based image reconstruction (MBIR) problem. They replace the proximal mapping with a learnable deep neural network in a layer-specific manner. In contrast, we construct the entire deep network, named DECUN, by unrolling and reparametrizing the HQS algorithm, and convergence of DECUN is guaranteed by the new parameterization scheme, as detailed in \eqref{eq:parameterization}. While both approaches provide convergence guarantees, theirs relies on technical assumptions that are somewhat challenging to verify prior to executing their algorithm (as these depend on the properties of critical points and fixed points). Our method, however, only requires mild technical assumptions that are generally valid in most practical scenarios. 2) Chun \emph{et al} \cite{Chun2023} do not characterize convergence rate, whereas we derive a detailed formula for the convergence rate in \eqref{convergence_rate_w}. 3) The application domain of \cite{Chun2023} is medical imaging, which differs from our focus on natural image deblurring.

The work in \cite{liu2019alista} is centered on establishing the analytical formula for optimal LISTA weights, showing that these weights offer a convergence rate comparable to those derived from training data. Their work focuses on the sparse coding problem and not focused on deconvolution.
\section{Experimental Validation}\label{sec:experiemnts}

In this section, we first provide numerical results on convergence verification as well as convergence rate calculation, which justifies the convergence guarantee of our proposed method. Then, we compare DECUN with state-of-the-art deblurring techniques through experiments, which illustrates DECUN's superior performance and computational benefits. Via cross-validation, we set the hyperparameter $\mu=5\times10^4$. As the training loss function to learn $\bar{\bD}$, $\bE^l$, $\bar{\beta}$, we use a linear combination of the Mean Squared Error (MSE) together with Mean Absolute Error (MAE), since MSE and MAE are two loss-terms widely used in non-blind image deconvolution literature  \cite{zhang2017learning, ren2018deep}. 
Though $l^1$ norm in MAE loss is non-differentiable at 0, MAE loss has been widely used in neural networks as a training loss function \cite{pandey2019new, zhang2019late}, and many modern machine learning frameworks like PyTorch, have already incorporated support for MAE loss in training \cite{pytorchL1}. In addition, in \cite{zhao2016loss} researchers have conducted thorough experiments to show that MAE loss provides superior results to differentiable alternatives such as MSE loss.
The Adam optimizer \cite{kingma2014adam} is employed to learn the DECUN network parameters. The code and trained models are available at \url{https://github.com/6zhc/DECUN}.

\subsection{Numerical simulation}\label{ssec:exp_simulation}
% In our experiments, we set $\mu = 5\times10^{4}$ and $\bar{\beta}=2\times 10^{3}$.
In order to verify the convergence of our proposed method, we carried out a set of experiments on a $256\times256$ image selected from \cite{Levin2011}. A motion blur kernel $k$ is chosen from \cite{Levin2011}, which is applied to the image with additive Gaussian noise $n$ with zero mean and standard deviation $10^{-5}$.  To further demonstrate the convergence of our proposed method, we define the error value as $\text{Error}= \frac{\|\bw^{l}-\bw^{\ast}\|}{\|\bw^{\ast}\|}$, and calculate how this value change with respect to iteration $l$ from 1 to 30, where $\|\bw^{\ast}\|$ is the fixed point obtained by executing 500 iterations. 
\par First of all, in order to illustrate the convergence property of our proposed method, we set the number of filters $C = 2$, and the parameters 
%$\bD^l=\bar{\bD}+\xi_l\bE^l,\quad\beta^l=\bar{\beta}+\gamma_l$, where $\bar{\bD}$ includes $2$ filters ($C=2$) with $\bar{\bD}_{1}=\bigl(\begin{smallmatrix}-1&1\\0&0\end{smallmatrix}\bigr)$ and $\bar{\bD}_{2}=\bigl(\begin{smallmatrix}-1&0\\1&0\end{smallmatrix}\bigr)$, $\bar{\beta}=2\times 10^{3}$ and $\bE^l=\bigl(\begin{smallmatrix}1&1\\1&1\end{smallmatrix}\bigr)$ are fixed, and 
$\xi_{l}$ and $\gamma_{l}$ as two forms of diminishing sequences: in one form, we take them to be exponential series, i.e., we choose $\xi_{l}=\xi^{l}$ and $\gamma_{l}=\gamma^{l}$ with $0<\xi<1$ and $0<\gamma<1$; in another form, we select the $p$-series as $\xi_{l}=\left(\frac{1}{l}\right)^{p}$ and $\gamma_{l}=\left(\frac{1}{l}\right)^{p}$ with $p\geq1$. In comparison, we choose random sequences for both of them, in which both $\xi_{l}$ and $\gamma_{l}$ satisfy Gaussian random distribution $\xi_{l}\sim \mathcal{N}(0,\sigma_{\xi_{l}}^{2})$ and $\gamma_{l}\sim \mathcal{N}(0,\sigma_{\beta_{l}}^{2})$, where $\sigma_{\xi_{l}}=\sigma_{\beta_{l}}=l/60$. The convergence comparison results are shown in Fig. \ref{Comparison}.
As can be seen from Fig. \ref{Comparison}, it is obvious that when the parameters are in converging forms (\emph{e.g.} $\xi_{l}=\xi^{l}$ and $\gamma_{l}=\gamma^{l}$ with $0<\xi<1$ and $0<\gamma<1$; and $\xi_{l}=\left(\frac{1}{l}\right)^{p}$ and $\gamma_{l}=\left(\frac{1}{l}\right)^{p}$ with $p\geq1$), the error value $\text{Error}= \frac{\|\bw^{l}-\bw^{\ast}\|}{\|\bw^{\ast}\|}$ is convergent with respect to the iteration $l$. In contrast, the error value $\text{Error}= \frac{\|\bw^{l}-\bw^{\ast}\|}{\|\bw^{\ast}\|}$ is divergent when the parameters $\xi_{l}$ and $\gamma_{l}$ are chosen to be divergent (\emph{e.g.} $\xi_{l}\sim \mathcal{N}(0,\sigma_{\xi_{l}}^{2})$ and $\gamma_{l}\sim \mathcal{N}(0,\sigma_{\beta_{l}}^{2})$), which confirms our claim that convergence is guaranteed when $\xi_{l}$ and $\gamma_{l}$ are both convergent, and the method is divergent otherwise. 

\begin{figure}[!htbp]
\centering
\includegraphics[width=2.5in]{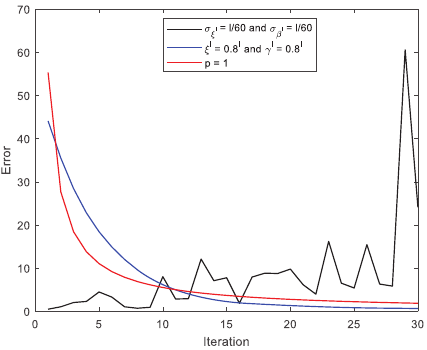}
\caption{Convergence comparison of our proposed method with three different types of parameter settings}
\label{Comparison}
%\vspace{-5mm}
\end{figure}
\par Next, we take the logarithm of $\text{Error}=\frac{\|\bw^{l}-\bw^{\ast}\|}{\|\bw^{\ast}\|}$ as $\log(\text{Error})$ to further speculate the convergence rate of our proposed method. The numerical results of both error values $\text{Error}=\frac{\|\bw^{l}-\bw^{\ast}\|}{\|\bw^{\ast}\|}$ and $\log(\text{Error})$ value are shown in Fig. \ref{Simulation1}, where the parameters $\xi_{l}$ and $\gamma_{l}$ are set to be the form $\xi_{l}=\xi^{l}$ and $\gamma_{l}=\gamma^{l}$ with $0<\xi<1$ and $0<\gamma<1$. As can be seen from Fig. \ref{Simulation1}. (a), when the parameter $\xi$ and $\gamma$ change from $0.2$ to $0.8$ with the step of $0.2$, both error values $\text{Error}=\frac{\|\bw^{l}-\bw^{\ast}\|}{\|\bw^{\ast}\|}$ and $\log(\text{Error})$ converge with respect to iteration $l$. Furthermore, Fig. \ref{Simulation1}(b) shows that the smaller the parameters $\xi$ and $\gamma$ are, 
the faster the convergence speed of error values $\text{Error}=\frac{\|\bw^{l}-\bw^{\ast}\|}{\|\bw^{\ast}\|}$ are.
\begin{figure}[!htbp]
\centering
\includegraphics[width=2.8in]{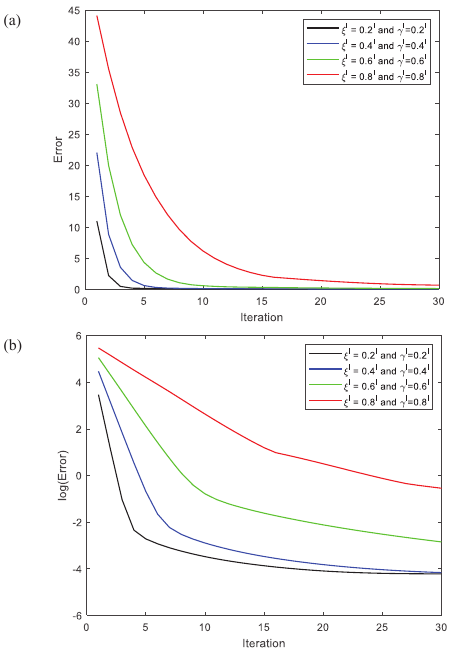}
\caption{(a) $\text{Error}=\frac{\|\bw^{l}-\bw^{\ast}\|}{\|\bw^{\ast}\|}$ with parameters $\xi=0.2,0.4,0.6~\text{and}~0.8$ and $\gamma=0.2,0.4,0.6~\text{and}~0.8$~(b) $\log(\text{Error})$ with parameters $\xi=0.2,0.4,0.6~\text{and}~0.8$ and $\gamma=0.2,0.4,0.6~\text{and}~0.8$}
\label{Simulation1}
%\vspace{-2mm}
\end{figure}
\par In addition, we also choose the parameters $\xi_{l}$ and $\gamma_{l}$ as the form $\xi_{l}=\left(\frac{1}{l+1}\right)^{p}$ and $\gamma_{l}=\left(\frac{1}{l+1}\right)^{p}$ with $p\geq1$ satisfying $\lim_{l\rightarrow+\infty}\left(\frac{1}{l+1}\right)^{p}=0$. To show convergence of our proposed method, we calculate both error values $\text{Error}=\frac{\|\bw^{l}-\bw^{\ast}\|}{\|\bw^{\ast}\|}$ and $\log(\text{Error})$ with $p=1,2,3~\text{and}~4$. The numerical results in this case are illustrated in Fig. \ref{Simulation2}. 
As shown in Fig. \ref{Simulation2}. (a), both error values $\text{Error}=\frac{\|\bw^{l}-\bw^{\ast}\|}{\|\bw^{\ast}\|}$ and $\log(\text{Error})$ are convergent with respect to iteration $l$ for $p=1,2,3~\text{and}~4$. The plot in Fig. \ref{Simulation2}~(b) clearly illustrates that the convergence speed of error value $\text{Error}=\frac{\|\bw^{l}-\bw^{\ast}\|}{\|\bw^{\ast}\|}$ is much faster for larger parameter $p$.
\begin{figure}[!htbp]
\centering
\includegraphics[width=2.8in]{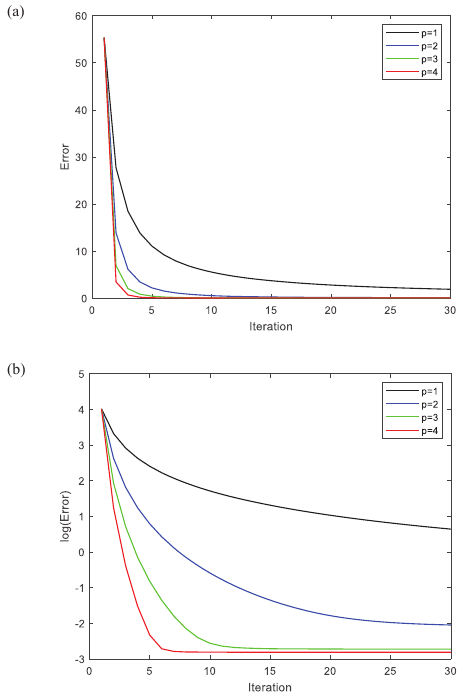}
\caption{(a) $\text{Error}=\frac{\|\bw^{l}-\bw^{\ast}\|}{\|\bw^{\ast}\|}$ with parameter $p=1,2,3~\text{and}~4$~(b) $\log(\text{Error})$ with parameter $p=1,2,3~\text{and}~4$}
\label{Simulation2}
\vspace{-2mm}
\end{figure}
%\subsection{Experiments}
%\section{Experimental Verification}
\subsection{Experiments on Benchmark Image Datasets}
\subsubsection{Training and Test Datasets Setup}
\begin{itemize}

% \item \textbf{Training Dataset for Ablation Study}
%  Microsoft Common Objects in Context (MS COCO) \cite{lin2014microsoft} is a large-scale detection, segmentation, key-point detection, and captioning dataset which consists of 328K images. We randomly selected 800 images from it and cropped them into $424\times424$ pixels. For blur kernel, we generated 10000 nonlinear kernels by recording camera motion trajectories with the method proposed in \cite{Li2020}. 

% \item \textbf{Testing dataset for Ablation Study}
%  200 images from MS COCO, which has no non-overlapping with training images, have been selected and cropped into $424\times424$ pixels, the same size as the training dataset. Also, 200 nonlinear kernels different from the training dataset were generated with the method proposed in \cite{Li2020}. 

\item \textbf{Training dataset for Linear Motion Kernels:}
Microsoft Common Objects in Context (MS COCO) \cite{lin2014microsoft} is a large-scale detection, segmentation, and captioning dataset which consists of 328K images. We randomly selected 800 images from it and cropped them into $424\times424$ pixels. 300 linear motion kernels are generated by randomly choosing lengths from 0 to 20 pixels and angles between 0 and $\pi$, followed by 2D spatial interpolation. Each clear image is convolved with every blur kernel (followed by the addition of Gaussian noise with 1\% standard deviation) resulting in $800\times300$ training image pairs. 
\item \textbf{Testing dataset for Linear Motion Kernels:}
Berkeley Segmentation Dataset 500 (BSD500), which consists of 500 natural images, was used to test the trained model. 100 images from BSD500 were chosen and 10 motion linear kernels that are non-overlapping with training kernels were applied along with noise addition to generate the $100 \times 10$ test images. 

\item \textbf{Training dataset for General Motion Kernels:}
As before, we selected 800 images from MS COCO and cropped them into $424\times424$ pixels. For blur kernels, we generated 10000 nonlinear kernels by recording camera motion trajectories with the method proposed in \cite{kohler2012recording}. For each step in the training stage, we randomly chose a batch of clear images and blur kernels from 800 clear images and 10000 nonlinear kernels, convolved each clear image with a blur kernel, and added Gaussian noise to the blurry images with 1\% noise standard deviation.

\item \textbf{Testing dataset for General Motion Kernels:}
 100 images from BSD500 were chosen along with 100 nonlinear kernels (non-overlapping with training). Furthermore, to demonstrate the versatility and generalizability of the model across different datasets, we also chose 300 images from the VOC2012\cite{pascal-voc-2012} along with another 300 nonlinear kernels and cropped them into $288 \times 456$ pixels.  Convolution with blur kernels followed by noise addition provides our test set.

\end{itemize}

Note that, we purposely use training and testing imagery from distinct datasets to examine how well competing methods perform, i.e. generalize, across datasets.

\subsubsection{Quantitative Performance Measures}

To quantitatively assess the performance of various methods in different scenarios, we use the Peak-Signal-to-Noise Ratio (PSNR) and Structural Similarity Index (SSIM) \cite{1284395}.

\subsubsection{Number of Parameters}

In our investigation, we examine the impact of varying layer counts on network performance. This analysis will be detailed further in the ablation study (Section IV.C). Referring to Table \ref{tab2}, we observe the network's behavior with differing layer quantities. Consequently, for the DECUN network's comparison with SOTA, we designate the number of layers $L$ as 30 and the number of filters $C$ as 4.
Regarding the number of parameters within the DECUN framework, the total count of trainable layer-specific filters across the entirety of the network amounts to $C \times L$. Moreover, the network comprises one additional trainable filter, the filter $\bar{\mathbf{D}}$, and a trainable parameter, the relaxation strength $\bar{\beta}$. Summing up, the comprehensive tally of trainable parameters in the DECUN network totals $(C \times L + C) \times S_{\text{filter}} + 1$, where $S_{\text{filter}}$ is the size of the filter. The DECUN framework, employed for SOTA analysis, features 4 trainable filters with the size of $3\times 3$ in each of its 30 layers. Therefore, the total number of parameters is calculated as $(30\times4 +4) \times (3\times3) +1 = 1117$. Compared with state-of-the-art methods e.g., MoDL \cite{aggarwal2018modl} with $188\text{K}$ parameters used, DWDN\cite{dong2021dwdn} deployed with $7.05\text{M}$ parameters and INFWIDE~\cite{zhang2023infwide} that uses $15.74\text{M}$ parameters, the proposed DECUN employs the fewest parameters, leading to significant computational savings.
\vspace{-2mm}
\subsection{Ablation Study}

To provide insights into the design of our network, we first carry out a series of experimental studies to investigate the influence of three key design factors: 1.) the filter sequence type among all the layers, and 2.) the number of layers $L$. %, and 3.) the number of filters $C$.
To train the models, 800 cropped images are chosen from MS COCO dataset and 10000 nonlinear kernels were generated with the method proposed in \cite{kohler2012recording}. For each step in the training stage, we randomly chose a batch of clear images and blur kernels, convolved each clear image with a blur kernel, and added Gaussian noise to the blurry images with a 1\% std. In the testing stage,  200 images from MS COCO and 200 nonlinear kernels non-overlapping with training were used for the ablation study results presented next. % We evaluate performance against the ground truth test images using PSNR and SSIM. 

We compare the proposed DECUN with traditional HQS and Unrolled Layer specific learning (U-LSL). All three methods are designed with 10 layers and 2 filters in each layer. In this case, the main difference between the three methods is the relationship of filters between each layer. For traditional HQS, each layer's filters are always horizontal and vertical image gradient filters and $\beta$ is fixed too. For DECUN, we follow the proposed parameterization in (\ref{eq:parameterization}) with $\xi_l=0.5^l$ and $\gamma_l=0.5^l$. Then U-LSL is the deep unrolled HQS network with independent layer-specific learning of filters (and $\beta^{l}$), which leads to the least restrictions. Table \ref{tab1} summarizes the numerical scores corresponding to these three methods. Both U-LSL and DECUN benefit from training data-inspired deep learning of parameters. 
U-LSL performs the best unsurprisingly as it affords the most modeling power. DECUN performs much better than the traditional HQS and only slightly worse than U-LSL. In all future experiments, we will stick to DECUN comparisons with state-of-the-art.

\begin{table}[htbp]
\caption{Ablation study showing benefits of deep unrolled networks -- convergent and otherwise over traditional HQS.}
\centering
\begin{tabular}{|c|c|c|c|}
\hline
 & \textbf{Traditional HQS} 
 % & \textbf{Same} 
   & \textbf{DECUN } 
 & \textbf{U-LSL} \\ 
\hline
SSIM & 0.8264%600454164458
    %& 0.9115%68371
    & 0.9177%47314
    & \textbf{0.9187}%33090
                 \\ 
                 \hline
PSNR & 26.4545%14979620836 
    %& 29.8605%4309
    & 30.3508%3925 
    & \textbf{30.4786}%3532
                \\ 
                \hline
\end{tabular}
\label{tab1}
\end{table}

We next study the effects of different numbers of layers $L$. The network performance over different choices of the network layer number is summarized in Table \ref{tab2} with the number of filters $C$ set to 4. As the network layer number increases, the network performance clearly improves but the rate of improvement decreases, which is consistent with the common observation that deeper networks tend to perform better. We observe diminishing performance gains after 30 layers. For results reported next, we set $L = 30$ and $C=4$.
%At the same time, we also test the network with different filter numbers. To maximize the number of filters within limited GPU memory, the layer number is set as 10. The network performance over different choices of filter number is summarized in Table \ref{tab3}. The network performs better when the filter number increases from 2 to 7. However, the score slightly comes down when the filter number is 8, presumably because the larger number of filters renders the network harder to train.

\begin{table}[htbp]
  \caption{DECUN performance w.r.t number of layers.}
\centering
        \begin{tabular}{|c|c|c|c|c|c|c|}
    \hline
    & \textbf{10 layers} & \textbf{15 layers} & \textbf{20 layers} & \textbf{25 layers} & \textbf{30 layers} 
    \\ \hline
    SSIM & 0.9141%79701
    & 0.9260%400812504059	
    & 0.9288%52317
    & 0.9298%893504528852
    & \textbf{0.9313}%03129
   \\ \hline
    PSNR & 30.3956%9046
    & 31.1737%80173754427
    & 31.4070%5864
    & 31.4660%05022298308
    & \textbf{31.6179}%4312              
    \\ \hline
    \end{tabular}
\label{tab2}
\end{table}

\begin{figure*}[t]
    \centering
    \begin{subfigure}[t]{0.084\textwidth}
        \centering
        \includegraphics[width=\textwidth,trim=22 22 22 22,clip]{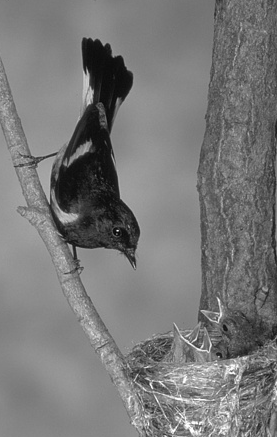}
        \includegraphics[width=\textwidth,trim=22 22 22 22,clip]{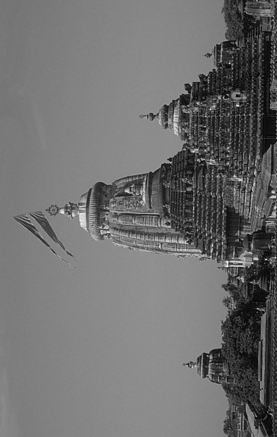}
        \caption{Clear}
    \end{subfigure}%
    \begin{subfigure}[t]{0.084\textwidth}
        \centering
        \includegraphics[width=\textwidth]{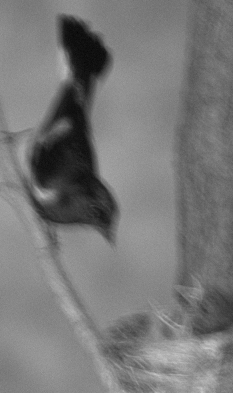}
        \includegraphics[width=\textwidth]{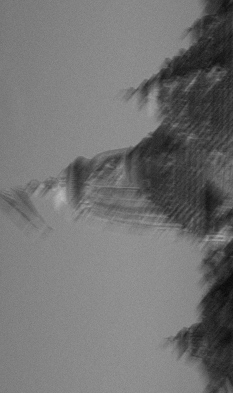}
        \caption{Blurred}
    \end{subfigure}%
    \centering
    \begin{subfigure}[t]{0.084\textwidth}
        \centering
        \includegraphics[width=\textwidth,trim=66 66 66 66,clip]{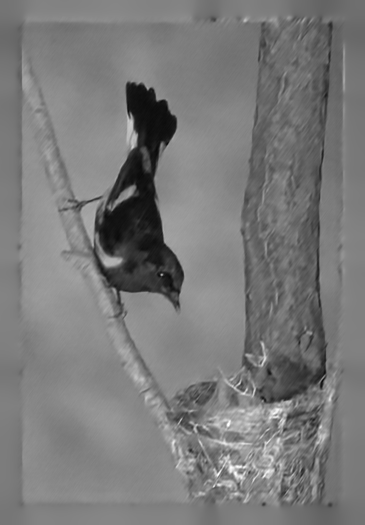}
        \includegraphics[width=\textwidth,trim=66 66 66 66,clip]{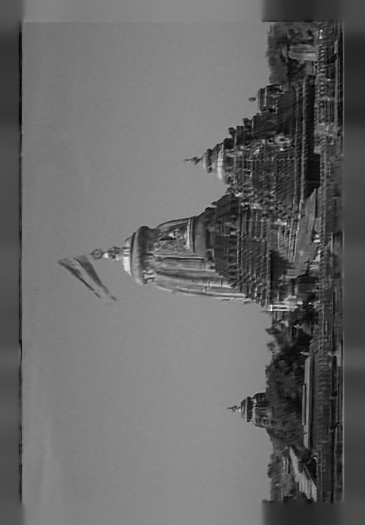}
        \caption{EPLL}
    \end{subfigure}%
    \begin{subfigure}[t]{0.084\textwidth}
        \centering
        \includegraphics[width=\textwidth,trim=22 22 22 22,clip]{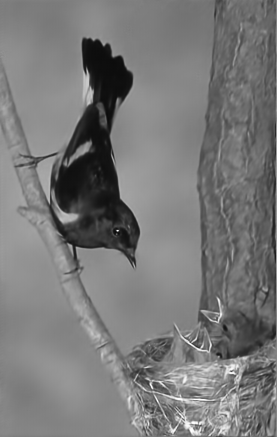}
        \includegraphics[width=\textwidth,trim=22 22 22 22,clip]{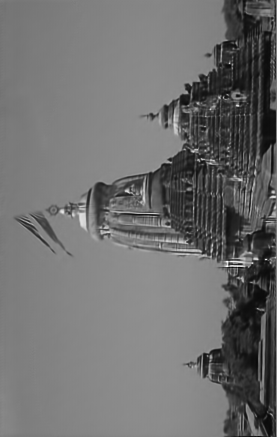}
        \caption{MLP}
    \end{subfigure}%
    \centering
    \begin{subfigure}[t]{0.084\textwidth}
        \centering
        \includegraphics[width=\textwidth,trim=44 44 44 44,clip]{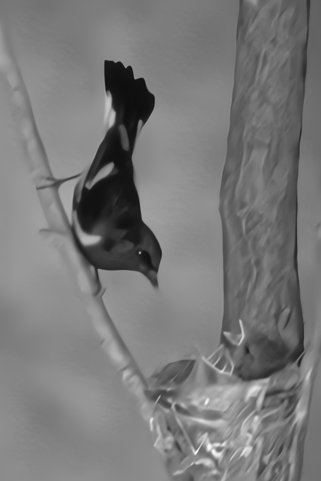}
        \includegraphics[width=\textwidth,trim=44 44 44 44,clip]{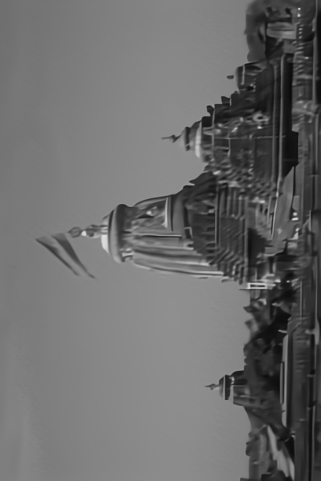}
        \caption{IRCNN}
    \end{subfigure}%
    \begin{subfigure}[t]{0.084\textwidth}
        \centering
        \includegraphics[width=\textwidth,trim=44 44 44 44,clip]{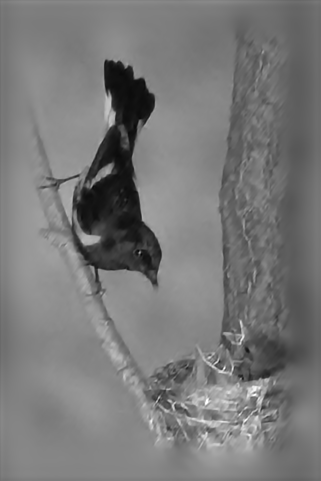}
        \includegraphics[width=\textwidth,trim=44 44 44 44,clip]{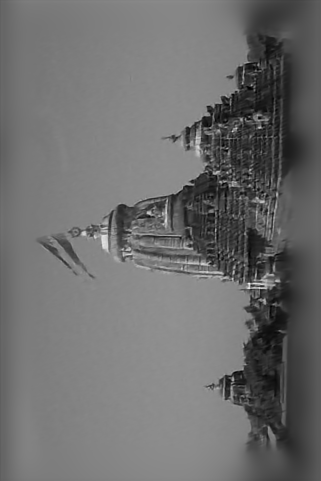}
        \caption{FCNN}
    \end{subfigure}%
    \centering
    \begin{subfigure}[t]{0.084\textwidth}
        \centering
        \includegraphics[width=\textwidth,trim=22 22 22 22,clip]{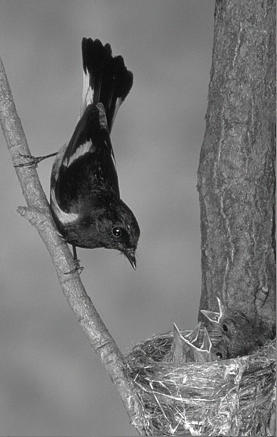}
        \includegraphics[width=\textwidth,trim=22 22 22 22,clip]{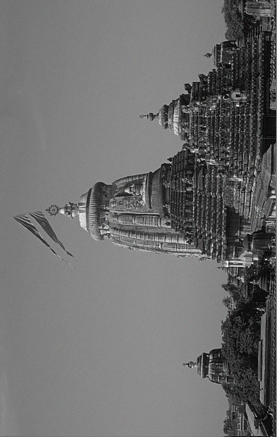}
        \caption{FNBD}
    \end{subfigure}%
    \begin{subfigure}[t]{0.084\textwidth}
        \centering
        \includegraphics[width=\textwidth,trim=44 44 44 44,clip]{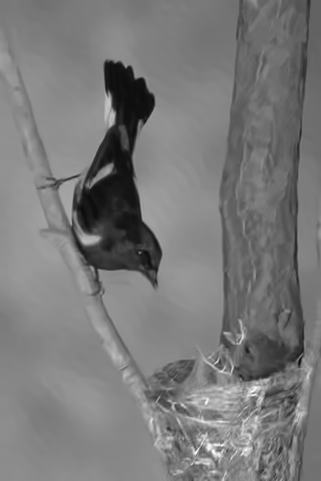}
        \includegraphics[width=\textwidth,trim=44 44 44 44,clip]{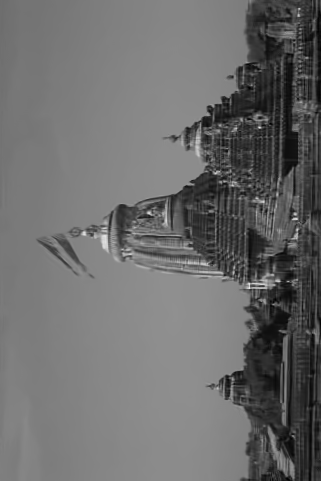}
        \caption{BM3D}
    \end{subfigure}%
    \begin{subfigure}[t]{0.084\textwidth}
        \centering
        \includegraphics[width=\textwidth,trim=22 22 22 22,clip]{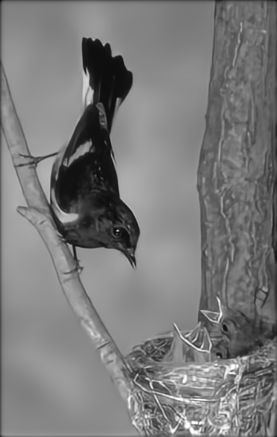}
        \includegraphics[width=\textwidth,trim=22 22 22 22,clip]{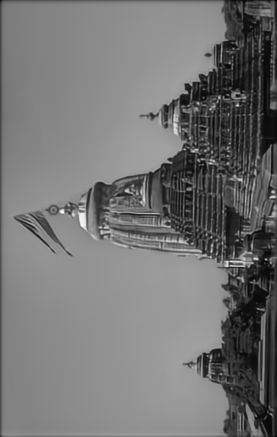}
        \caption{INFWIDE}
    \end{subfigure}%
    \centering
    \begin{subfigure}[t]{0.084\textwidth}
        \centering
        \includegraphics[width=\textwidth,trim=44 44 44 44,clip]{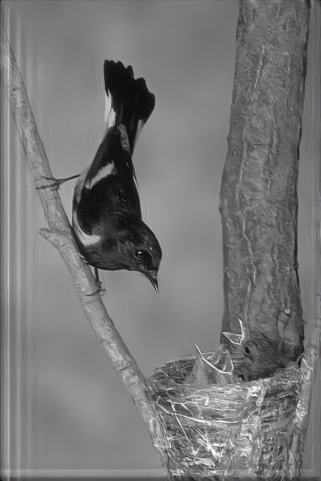}
        \includegraphics[width=\textwidth,trim=44 44 44 44,clip]{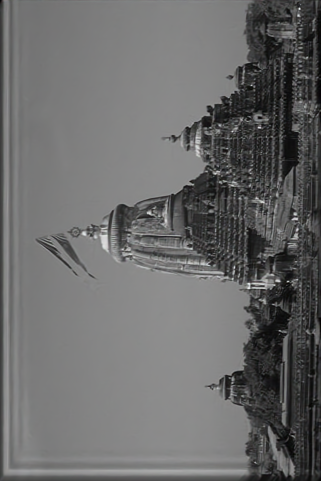}
        \caption{SVMAP}
    \end{subfigure}%
    \begin{subfigure}[t]{0.084\textwidth}
        \centering
        \includegraphics[width=\textwidth,trim=22 22 22 22,clip]{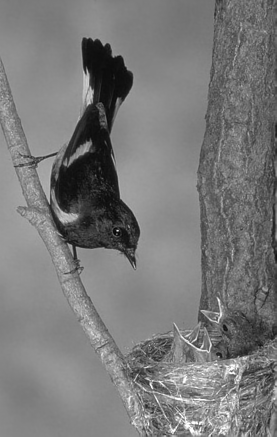}
        \includegraphics[width=\textwidth,trim=22 22 22 22,clip]{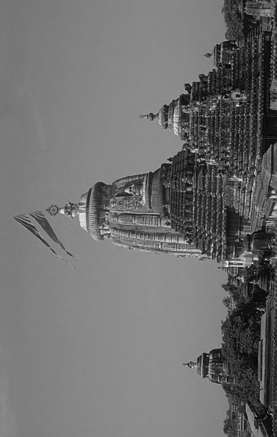}
        \caption{DWDN}
    \end{subfigure}%
    \begin{subfigure}[t]{0.084\textwidth}
        \centering
        \includegraphics[width=\textwidth,trim=44 44 44 44,clip]{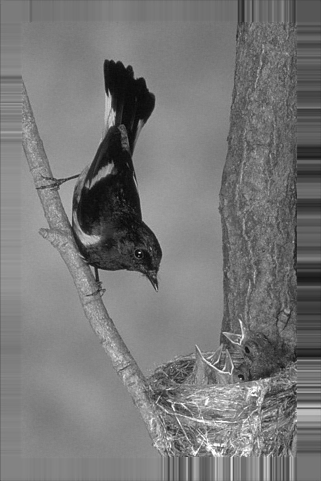}
        \includegraphics[width=\textwidth,trim=44 44 44 44,clip]{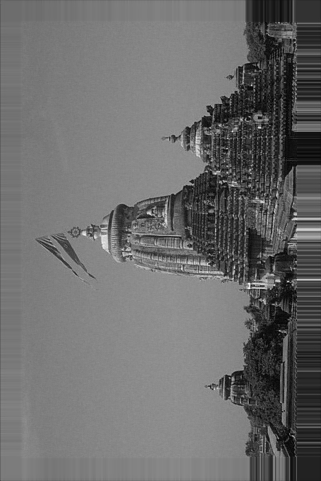}
        \caption{DECUN}
    \end{subfigure}

\caption{Visual comparisons over BSD dataset using linear kernels.}  % between our method and the state-of-the-art models based on BSD dataset.
\label{fig5}
\end{figure*}
\vspace{-2mm}
\subsection{Comparisons Against State of the Art}

We select ten state-of-the-art methods to compare against. Among them, EPLL~\cite{zoran2011learning}, BM3D~\cite{makinen2020collaborative}, SVMAP~\cite{dong2021learning}, INFWIDE~\cite{zhang2023infwide}, MoDL \cite{aggarwal2018modl}, and DWDN\cite{dong2021dwdn} are iterative algorithms, some of which have learned components. MLP~\cite{schuler2013machine}, IRCNN~\cite{zhang2017}, FCNN~\cite{zhang2017learning}, and FNBD~\cite{son2017fast} are fully deep learning methods. The iterative algorithms are interpretable, and convergent but suffer from long execution time. In contrast, the deep learning methods could potentially achieve superior performance in image quality measures (when training is sufficient) but lack interpretability. We will demonstrate that DECUN combines the advantages of both categories. 
% In the following experiments, we set the number of layers $L$ to 30 and the number of filters $C$ to 4, to get the best performance of DECUN with the limited GPU memory. 

\subsubsection{Evaluation on Trajectory Linear Kernels}

Fig \ref{fig5} shows sample test images. The scores of various methods are presented in Table \ref{tab4}. DECUN emerges overall as the best method, as evidenced by its faithful preservation of high-frequency textures and details. EPLL uses patch level  Maximum a-Posteriori estimates to reconstruct images from learned natural image patch statistics, which is time-consuming but performs the best among all the iterative algorithms. BM3D emerges as a generic denoising method and is later adopted into the deblurring domain. However, the collaborative filtering process might eliminate important high-frequency details in the images.
SVMAP uses the neural network to predict per-pixel spatially-variant features and, together with the maximum a-posteriori framework, keeps the most texture feature and has the best performance among all the other states of the art. Among deep learning-based approaches, FCNN performs second best, and FNBD performs mildly worse than FCNN, as they share similar ideas. Specifically, FCNN uses fast convolutional neural networks to remove the artifacts produced by the Wiener filter, whereas FNBD does it via a residual network with long/short skip connections. MLP uses a regularized inversion of the blur kernel in the Fourier domain as a first step and then applies a denoising network to suppress the high-frequency artifacts. %However, the denoising network might fail to generalize once the image noise follows a different distribution than the training data.
IRCNN performs denoising before feeding the images into a deblurring optimization method, which may potentially result in loss of textures. INFWIDE\cite{zhang2023infwide} introduces a dual-path structure that specifically eliminates noise and reconstructs overexposed areas within the visual domain, but for images that only have Gaussian noise added, this structure may not be able to fully exhibit its entire range of capabilities. DWDN \cite{dong2021dwdn} employed a U-Net structure refinement module to predict the deblurred image from the deconvolved deep features. This approach necessitated that the dimensions of the image be multiples of 8. To adapt DWDN for use with the BSD dataset, which has images of dimensions $321 \times 481$ pixels, we implemented a sliding window technique. Although DWDN demonstrates exceptional performance in image deblurring, the application of the sliding window method resulted in visible segmentation boundaries in some of the deblurred images.

\begin{table*}[ht]
\caption{Quantitative comparison over BSD dataset and linear kernels. The best score is in bold fonts.}
\centering
\scalebox{0.91}{
\begin{tabular}{|c|c|c|c|c|c|c|c|c|c|c|}
    \hline
    &\textbf{EPLL}\cite{zoran2011learning} & \textbf{MLP}\cite{schuler2013machine} & \textbf{IRCNN}\cite{zhang2017} & \textbf{FCNN}\cite{zhang2017learning} & \textbf{FNBD}\cite{son2017fast}   & \textbf{BM3D} \cite{makinen2020collaborative} & \textbf{INFWIDE}\cite{zhang2023infwide} & \textbf{SVMAP}\cite{dong2021learning} & \textbf{DWDN}\cite{dong2021dwdn} &  
    \textbf{DECUN} \\ \hline
    SSIM 
    & 0.8644%56126
    & 0.8268%38316
    & 0.8460%02148
    & 0.9105%28293
    & 0.9003%131
    & 0.8532%664
    & 0.8700%10411
    & 0.9155%670082809088 
    & 0.9117%716711832278
    & \textbf{0.9169}%68397
          \\ \hline
    PSNR 
    & 28.0313%5873
    & 26.6674%2391
    & 27.6833%3153
    & 30.9068%2908
    & 31.2710%2507
    & 29.0598%0906
    & 26.3423%657
    & 31.8671% 12913276415 
    & 31.6972% 10710674128
    & \textbf{33.2427}%2099
        \\ \hline
    \end{tabular}}
\label{tab4}
\end{table*}

\subsubsection{Evaluation on Trajectory Nonlinear Kernels}

Sample test images and visual examples of reconstructed images from different methods are shown in Fig \ref{fig6}. Scores of all the methods are summarized in Table \ref{tab5}. The deblurring problem on nonlinear kernels is much more challenging than linear kernels due to their significantly higher diversity. Therefore, the overall scores on nonlinear kernels are typically lower than those of linear kernels but generally follow the same trend. Again DECUN outperforms all competing methods and better reconstructs high-frequency details, followed by FCNN, and EPLL ranks as the top-performing iterative algorithm. Among other methods, SVMAP preserves textures well thanks to its per-pixel spatially-variant features prediction, but still underperforms DECUN. 
Since the quantitative score of EPLL, FCNN, SVMAP, and DECUN is close, we zoomed in on the elephant region and compared four processed regions in Fig \ref{fig10}. Due to patch-level processing, there are clear boundaries between patch and patch in the images processed by the EPLL. The image processed by FCNN may lack contrast in some regions, such as the grass in the elephant feet region. SVMAP processed the images to the smoothest image among all four images, but it also led to losing too much detail. Overall the image processed by DECUN is closest to the sharp image and performs best in detail among all the methods. DWDN\cite{dong2021dwdn} use multi-scale feature refinement component to progressively restore fine-scale detail from the deconvolved features, which required the dimensions of the input image to be multiples of 8. Therefore, we cropped the image into  $288 \times 456$ pixels and compared the SOTA including the DWDN with DECUN. Visual representations of the reconstructed images, originating from blurred versions created using the VOC2012 Dataset, can be viewed in Fig \ref{fig11}. The score of all methods over the VOC dataset and nonlinear kernels is detailed in Table \ref{tab7}. The multi-scale feature refinement module enables DWDN to restore image details and sharpen text, achieving good performance in metrics. However, noticeable artifacts appear on some images, as illustrated by the first row of images in Fig \ref{fig11}. Furthermore, we enhance the SOTA by incorporating the MoDL\cite{aggarwal2018modl} framework. This is achieved by adapting the measurement operator $\textbf{A}$, which converts image data into the measurement domain in \cite{aggarwal2018modl}, into a Toeplitz matrix $\bK$ that aligns with the convolution kernel $\bk$, based on \eqref{eq:linear_model}.  From both visual cues and quantitative scores, it's evident that DECUN consistently aligns closest to the original sharp images, standing out as the method that most faithfully preserves detail across varied datasets.

\begin{figure*}[htbp]
    \centering
    \begin{subfigure}[t]{0.084\textwidth}
        \centering
        \includegraphics[width=\textwidth,trim=22 22 22 22,clip]{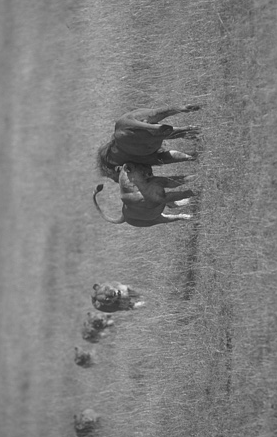}
        \includegraphics[width=\textwidth,trim=22 22 22 22,clip]{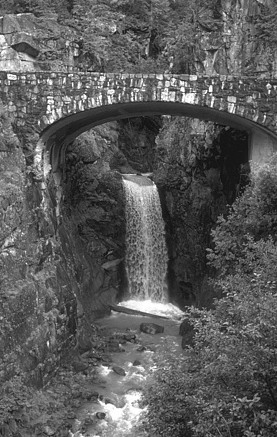}
        \caption{Clear}
    \end{subfigure}%
    \begin{subfigure}[t]{0.084\textwidth}
        \centering
        \includegraphics[width=\textwidth]{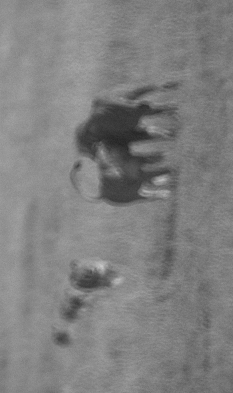}
        \includegraphics[width=\textwidth]{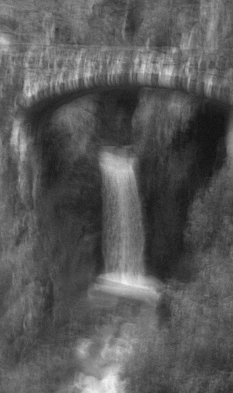}
        \caption{Blurred}
    \end{subfigure}%
    \centering
    \begin{subfigure}[t]{0.084\textwidth}
        \centering
        \includegraphics[width=\textwidth,trim=66 66 66 66,clip]{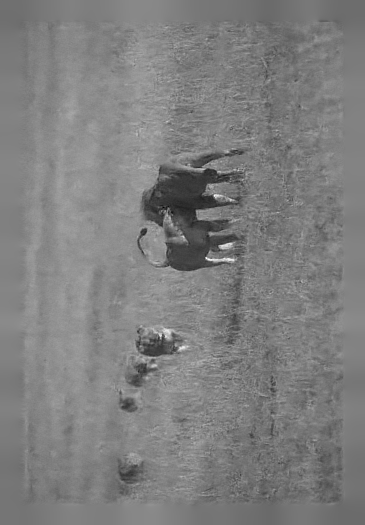}
        \includegraphics[width=\textwidth,trim=66 66 66 66,clip]{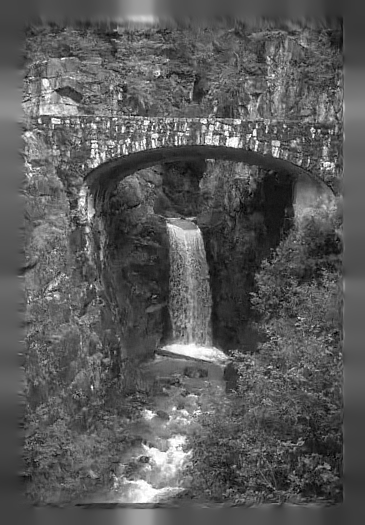}
        \caption{EPLL}
    \end{subfigure}%
    \begin{subfigure}[t]{0.084\textwidth}
        \centering
        \includegraphics[width=\textwidth,trim=22 22 22 22,clip]{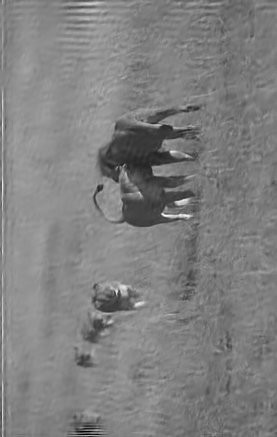}
        \includegraphics[width=\textwidth,trim=22 22 22 22,clip]{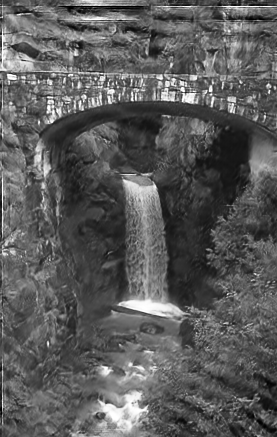}
        \caption{MLP}
    \end{subfigure}%
    \centering
    \begin{subfigure}[t]{0.084\textwidth}
        \centering
        \includegraphics[width=\textwidth,trim=44 44 44 44,clip]{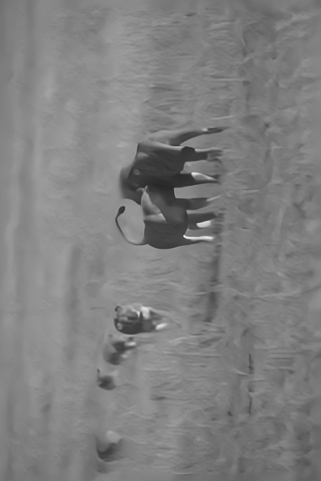}
        \includegraphics[width=\textwidth,trim=44 44 44 44,clip]{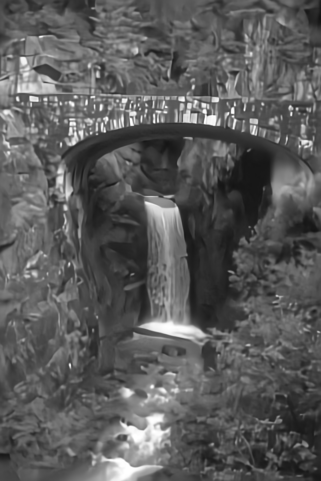}
        \caption{IRCNN}
    \end{subfigure}%
    \begin{subfigure}[t]{0.084\textwidth}
        \centering
        \includegraphics[width=\textwidth,trim=44 44 44 44,clip]{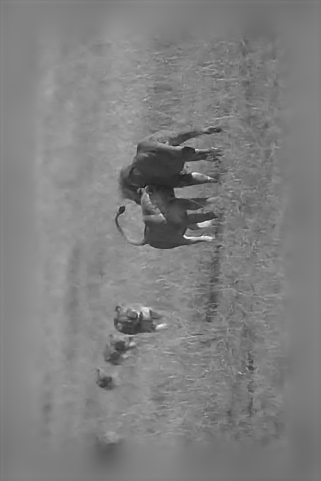}
        \includegraphics[width=\textwidth,trim=44 44 44 44,clip]{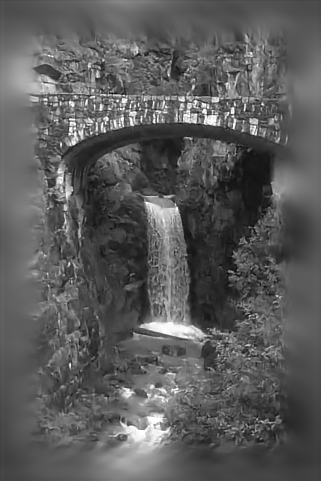}
        \caption{FCNN}
    \end{subfigure}%
    \centering
    \begin{subfigure}[t]{0.084\textwidth}
        \centering
        \includegraphics[width=\textwidth,trim=22 22 22 22,clip]{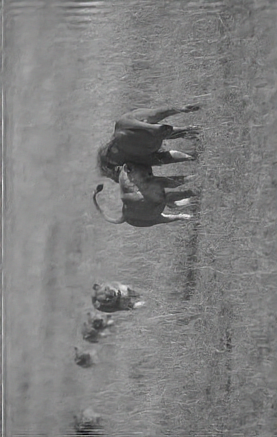}
        \includegraphics[width=\textwidth,trim=22 22 22 22,clip]{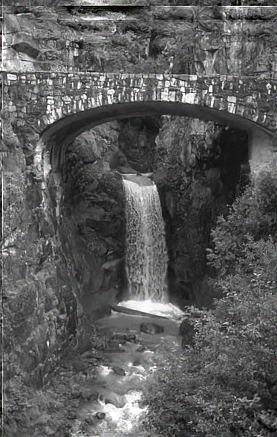}
        \caption{FNBD}
    \end{subfigure}%
    \begin{subfigure}[t]{0.084\textwidth}
        \centering
        \includegraphics[width=\textwidth,trim=44 44 44 44,clip]{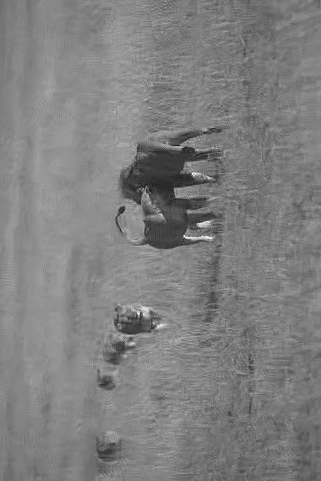}
        \includegraphics[width=\textwidth,trim=44 44 44 44,clip]{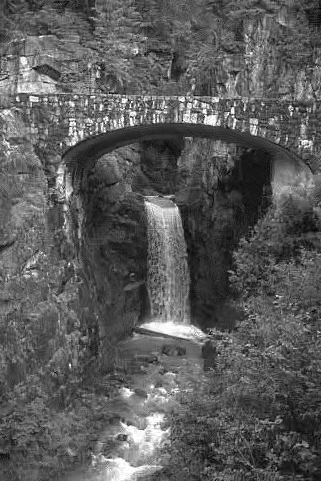}
        \caption{BM3D}
    \end{subfigure}%
    \begin{subfigure}[t]{0.084\textwidth}
        \centering
        \includegraphics[width=\textwidth,trim=22 22 22 22,clip]{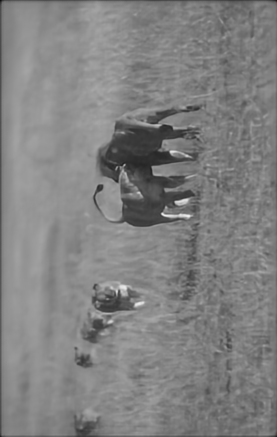}
        \includegraphics[width=\textwidth,trim=22 22 22 22,clip]{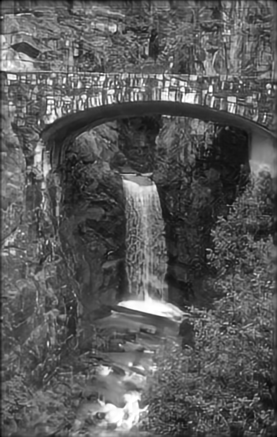}
        \caption{INFWIDE}
    \end{subfigure}%
    \centering
    \begin{subfigure}[t]{0.084\textwidth}
        \centering
        \includegraphics[width=\textwidth,trim=44 44 44 44,clip]{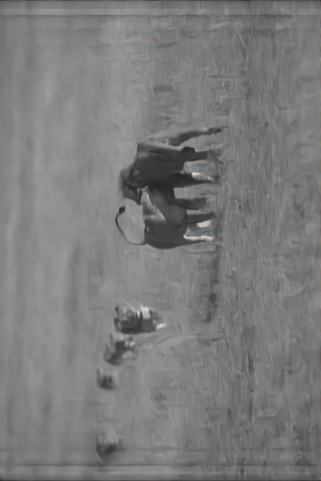}
        \includegraphics[width=\textwidth,trim=44 44 44 44,clip]{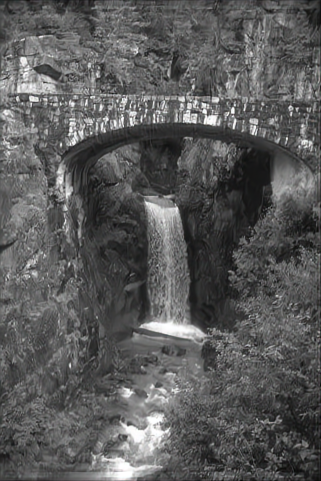}
        \caption{SVMAP}
    \end{subfigure}%
    \begin{subfigure}[t]{0.084\textwidth}
        \centering
        \includegraphics[width=\textwidth,trim=22 22 22 22,clip]{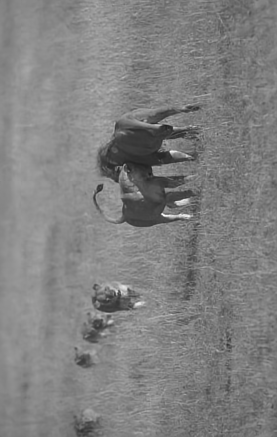}
        \includegraphics[width=\textwidth,trim=22 22 22 22,clip]{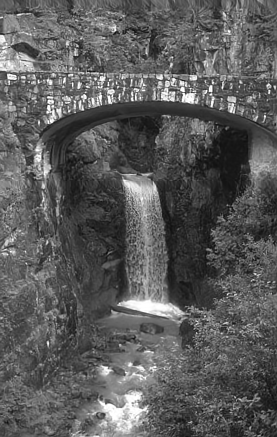}
        \caption{DWDN}
    \end{subfigure}%
    \begin{subfigure}[t]{0.084\textwidth}
        \centering
        \includegraphics[width=\textwidth,trim=44 44 44 44,clip]{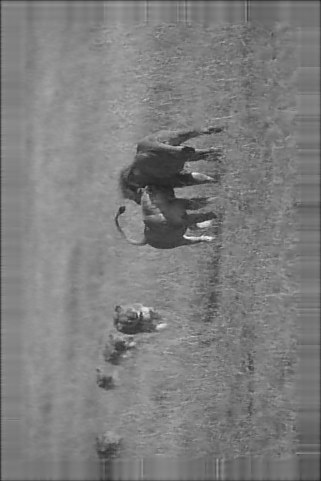}
        \includegraphics[width=\textwidth,trim=44 44 44 44,clip]{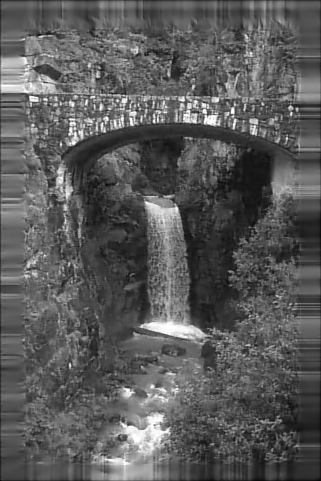}
        \caption{DECUN}
    \end{subfigure}  
\caption{Visual comparison over BSD dataset and nonlinear kernels.}  % between our method and the state-of-the-art models based on BSD dataset.

\label{fig6}
\end{figure*}

\begin{table*}[htbp]
\caption{Quantitative comparison over BSD dataset and nonlinear kernels. The best score is in bold fonts.}
\centering
\scalebox{0.91}{
\begin{tabular}{|c|c|c|c|c|c|c|c|c|c|c|}
    \hline
    & \textbf{EPLL}\cite{zoran2011learning} & \textbf{MLP}\cite{schuler2013machine} & \textbf{IRCNN}\cite{zhang2017} & \textbf{FCNN}\cite{zhang2017learning} & \textbf{FNBD}\cite{son2017fast}   & \textbf{BM3D} \cite{makinen2020collaborative} & \textbf{INFWIDE}\cite{zhang2023infwide} & \textbf{SVMAP}\cite{dong2021learning} &  
    \textbf{DWDN}\cite{dong2021dwdn} &  
    \textbf{DECUN} \\ \hline
    SSIM 
    & 0.8812%37336
    & 0.8329%37406
    & 0.8309%76285
    & 0.8898%65049
    & 0.8793%14806
    & 0.8586%3146
    & 0.8782%74058
    & 0.8880%811752989907
    & 0.8903%635783788847	
    & \textbf{0.8985}%30026
        \\ \hline
    PSNR 
    & 29.3769%7497
    & 26.9809%0488
    & 28.3595%3354
    & 29.8539%611
    & 29.5855%1929
    & 28.5916%4455
    & 28.3830%2493
    & 28.7086%21334789264
    & 30.2208%98445579635
    & \textbf{30.2715}%078
        \\ \hline
    \end{tabular}}
\label{tab5}
\end{table*}

% 60 180 50 90
\begin{figure*}[htbp]
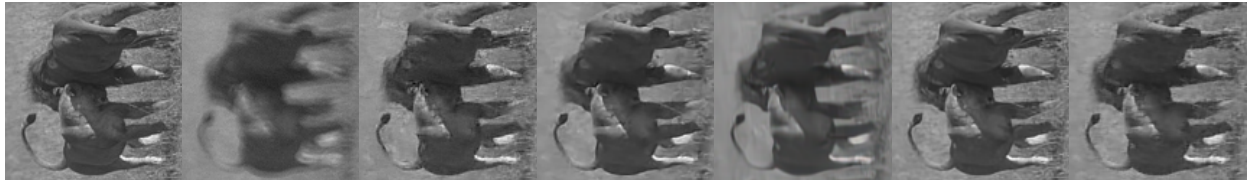

    \centering
    \begin{subfigure}[t]{0.13\textwidth}
        \centering
        \includegraphics[width=\textwidth,trim=82 207 72 107,clip]{BSD/x_01.png}
        \caption{Clear}
    \end{subfigure}%
    \begin{subfigure}[t]{0.13\textwidth}
        \centering
        \includegraphics[width=\textwidth,trim=60 185 50 85,clip]{BSD/y_01.png}
        \caption{Blurred}
    \end{subfigure}%
    \centering
    \begin{subfigure}[t]{0.13\textwidth}
        \centering
        \includegraphics[width=\textwidth,trim=126 251 116 151,clip]{BSD/EPLL_01.png}
        \caption{EPLL}
    \end{subfigure}%
    % \begin{subfigure}[t]{0.15\textwidth}
    %     \centering
    %     \includegraphics[width=\textwidth,trim=82 207 72 107,clip]{figs/experiment_result/BSD/MLP_01.png}
    %     \caption{MLP}
    % \end{subfigure}%
    % \centering
    % \begin{subfigure}[t]{0.15\textwidth}
    %     \centering
    %     \includegraphics[width=\textwidth,trim=104 229 94 129,clip]{figs/experiment_result/BSD/IRCNN_01.png}
    %     \caption{IRCNN}
    % \end{subfigure}%
    \begin{subfigure}[t]{0.13\textwidth}
        \centering
        \includegraphics[width=\textwidth,trim=104 229 94 129,clip]{BSD/FCNN_01.png}
        \caption{FCNN}
    \end{subfigure}%
    \centering
    % \begin{subfigure}[t]{0.15\textwidth}
    %     \centering
    %     \includegraphics[width=\textwidth,trim=82 207 72 107,clip]{figs/experiment_result/BSD/FNBD_01.png}
    %     \caption{FNBD}
    % \end{subfigure}%
    % \begin{subfigure}[t]{0.15\textwidth}
    %     \centering
    %     \includegraphics[width=\textwidth,trim=104 229 94 129,clip]{figs/experiment_result/BSD/BM3D_01.png}
    %     \caption{BM3D}
    % \end{subfigure}%
    \centering
    \begin{subfigure}[t]{0.13\textwidth}
        \centering
        \includegraphics[width=\textwidth,trim=104 229 94 129,clip]{BSD/SVMAP_01.png}
        \caption{SVMAP}
    \end{subfigure}%
    \begin{subfigure}[t]{0.13\textwidth}
        \centering
        \includegraphics[width=\textwidth,trim=82 207 72 107,clip]{BSD/DWDN_01.png}
        \caption{DWDN}
    \end{subfigure}%
    \begin{subfigure}[t]{0.13\textwidth}
        \centering
        \includegraphics[width=\textwidth,trim=104 229 94 129,clip]{BSD/DecovNetConvergenceBSDNonlinear0.01_01.png}
        \caption{DECUN}
    \end{subfigure}  
\caption{Visual comparison among the top four highest scoring methods.}  % between our method and the state-of-the-art models based on BSD dataset.
\label{fig10}
\end{figure*}

\begin{figure*}[htbp]
    \centering
    \begin{subfigure}[t]{0.090\textwidth}
        \centering
        \includegraphics[width=\textwidth,trim=22 22 22 22,clip]{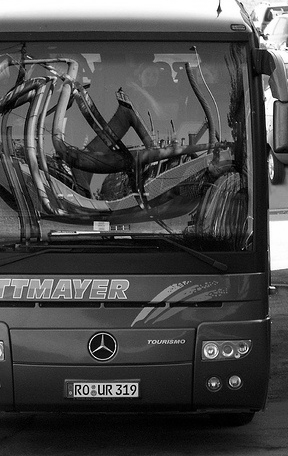}
        \includegraphics[width=\textwidth,trim=22 22 22 22,clip]{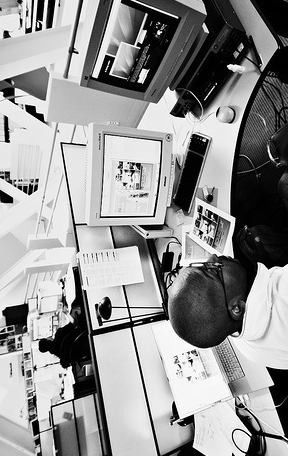}
        \caption{Clear}
    \end{subfigure}%
    \begin{subfigure}[t]{0.090\textwidth}
        \centering
        \includegraphics[width=\textwidth]{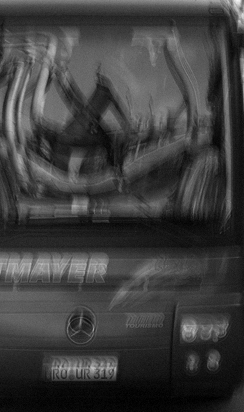}
        \includegraphics[width=\textwidth]{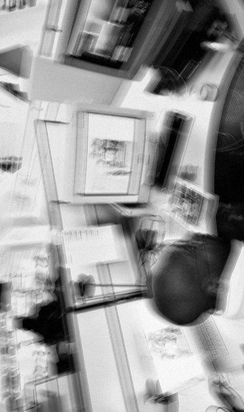}
        \caption{Blurred}
    \end{subfigure}%
    \centering
    \begin{subfigure}[t]{0.090\textwidth}
        \centering
        \includegraphics[width=\textwidth,trim=44 44 44 44,clip]{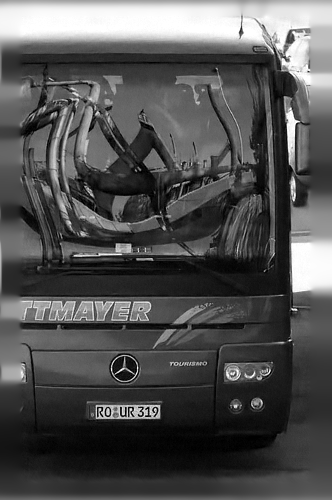}
        \includegraphics[width=\textwidth,trim=44 44 44 44,clip]{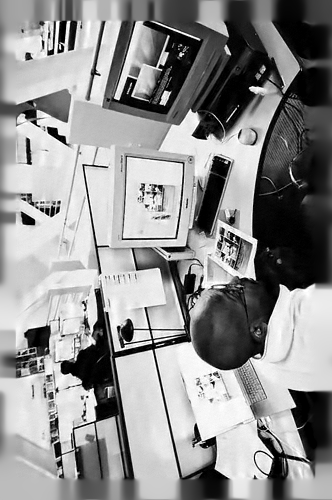}
        \caption{EPLL}
    \end{subfigure}%
    \begin{subfigure}[t]{0.090\textwidth}
        \centering
        \includegraphics[width=\textwidth,trim=44 44 44 44,clip]{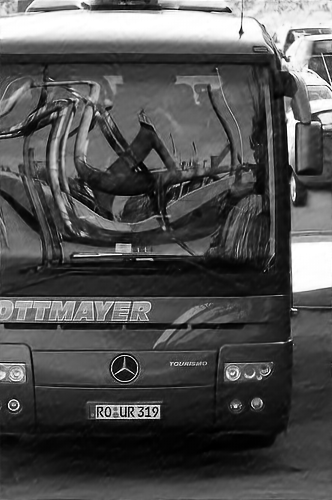}
        \includegraphics[width=\textwidth,trim=44 44 44 44,clip]{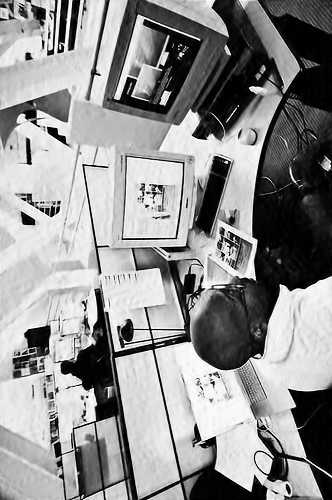}
        \caption{MLP}
    \end{subfigure}%
    \centering
    \begin{subfigure}[t]{0.090\textwidth}
        \centering
        \includegraphics[width=\textwidth,trim=44 44 44 44,clip]{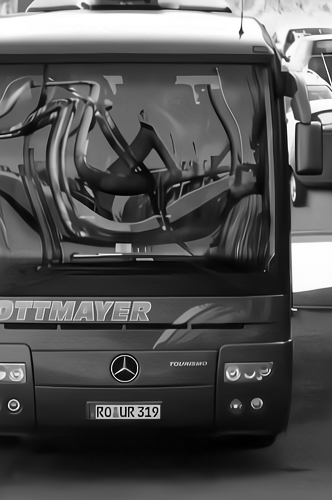}
        \includegraphics[width=\textwidth,trim=44 44 44 44,clip]{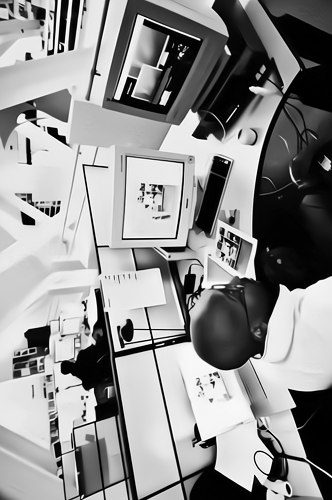}
        \caption{IRCNN}
    \end{subfigure}%
    \begin{subfigure}[t]{0.090\textwidth}
        \centering
        \includegraphics[width=\textwidth,trim=44 44 44 44,clip]{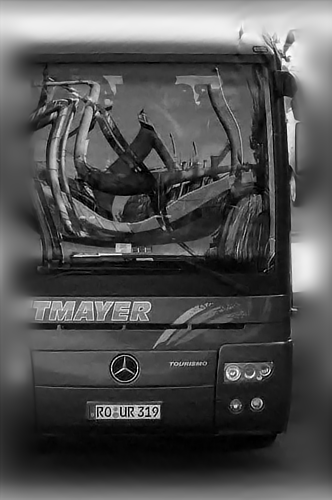}
        \includegraphics[width=\textwidth,trim=44 44 44 44,clip]{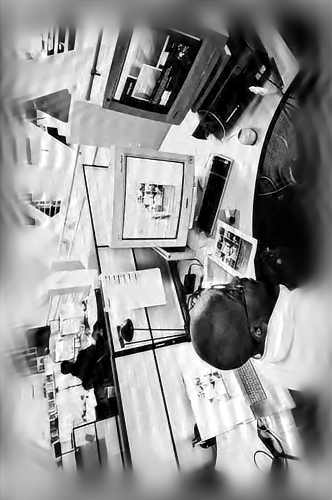}
        \caption{FCNN}
    \end{subfigure}%
    \centering
    \begin{subfigure}[t]{0.090\textwidth}
        \centering
        \includegraphics[width=\textwidth,trim=44 44 44 44,clip]{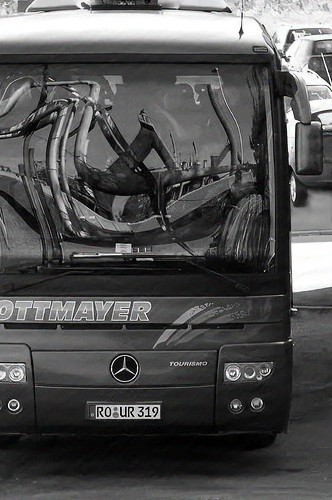}
        \includegraphics[width=\textwidth,trim=44 44 44 44,clip]{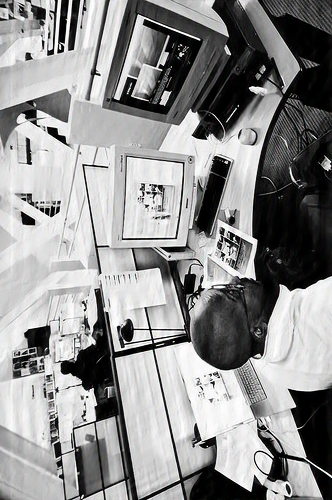}
        \caption{FNBD}
    \end{subfigure}%
    \begin{subfigure}[t]{0.090\textwidth}
        \centering
        \includegraphics[width=\textwidth,trim=44 44 44 44,clip]{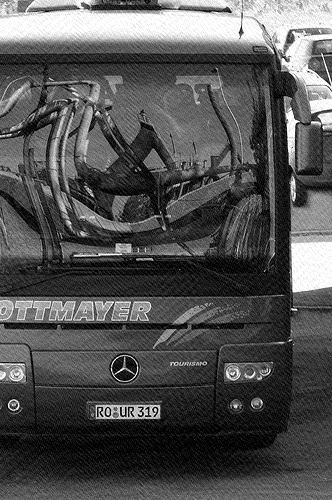}
        \includegraphics[width=\textwidth,trim=44 44 44 44,clip]{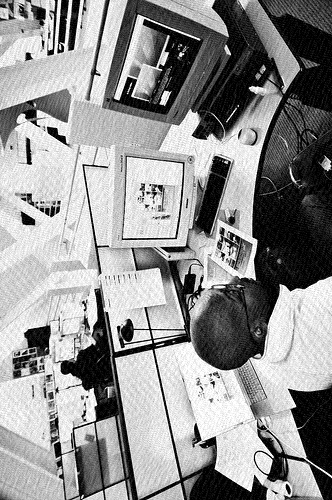}
        \caption{MoDL}
    \end{subfigure}%
    \begin{subfigure}[t]{0.090\textwidth}
        \centering
        \includegraphics[width=\textwidth,trim=44 44 44 44,clip]{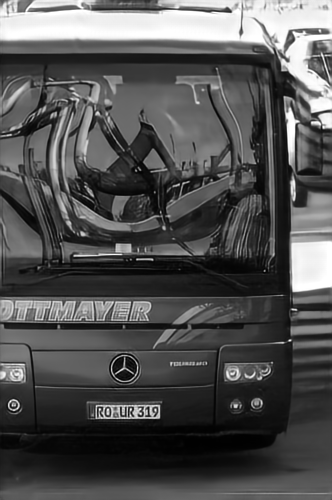}
        \includegraphics[width=\textwidth,trim=44 44 44 44,clip]{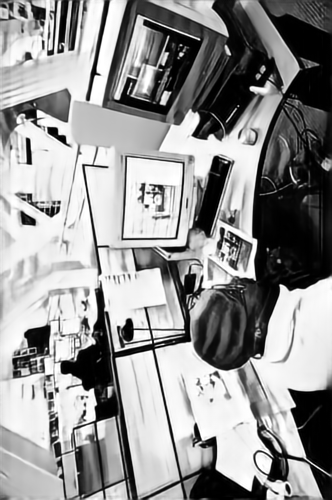}
        \caption{INFWIDE}
    \end{subfigure}%
    \centering
    % \begin{subfigure}[t]{0.090\textwidth}
    %     \centering
    %     \includegraphics[width=\textwidth,trim=44 44 44 44,clip]{figs/experiment_result/VOC/BM3D_01.png}
    %     \includegraphics[width=\textwidth,trim=44 44 44 44,clip]{figs/experiment_result/VOC/BM3D_02.png}
    %     \caption{BM3D}
    % \end{subfigure}%
    % \begin{subfigure}[t]{0.090\textwidth}
    %     \centering
    %     \includegraphics[width=\textwidth,trim=44 44 44 44,clip]{figs/experiment_result/VOC/SVMAP_01.png}
    %     \includegraphics[width=\textwidth,trim=44 44 44 44,clip]{figs/experiment_result/VOC/SVMAP_02.png}
    %     \caption{SVMAP}
    % \end{subfigure}%    
    \begin{subfigure}[t]{0.090\textwidth}
        \centering
        \includegraphics[width=\textwidth,trim=22 22 22 22,clip]{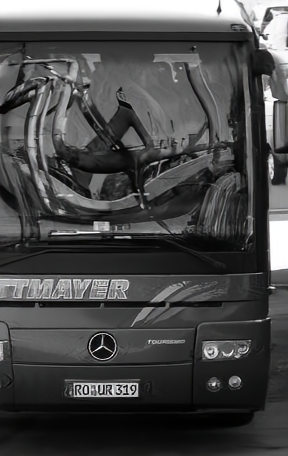}
        \includegraphics[width=\textwidth,trim=22 22 22 22,clip]{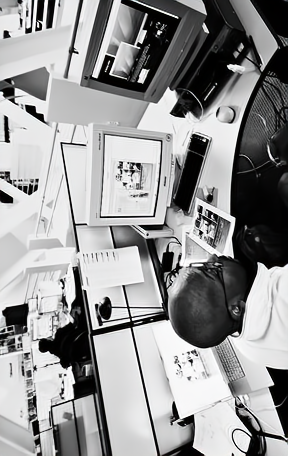}
        \caption{DWDN}
    \end{subfigure}%
    \begin{subfigure}[t]{0.090\textwidth}
        \centering
        \includegraphics[width=\textwidth,trim=44 44 44 44,clip]{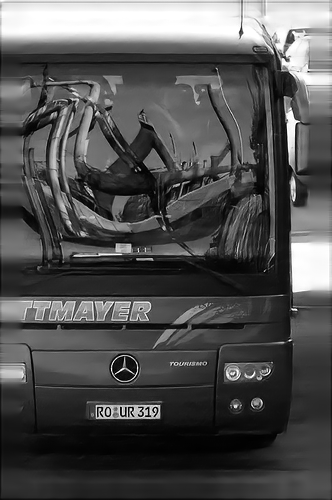}
        \includegraphics[width=\textwidth,trim=44 44 44 44,clip]{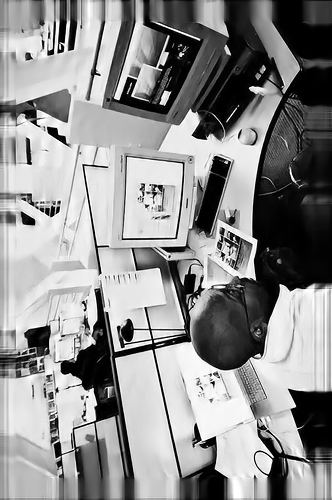}
        \caption{DECUN}
    \end{subfigure}  
\caption{Visual comparison over VOC dataset and nonlinear kernels.}  % between our method and the state-of-the-art models based on VOC dataset.

\label{fig11}
\end{figure*}

\begin{table*}[htbp]
\caption{Quantitative comparison over VOC dataset and nonlinear kernels. The best score is in bold fonts.}
\centering
\begin{tabular}{|c|c|c|c|c|c|c|c|c|c|}
    \hline
    & \textbf{EPLL}\cite{zoran2011learning} & \textbf{MLP}\cite{schuler2013machine} & \textbf{IRCNN}\cite{zhang2017} & \textbf{FCNN}\cite{zhang2017learning} & \textbf{FNBD}\cite{son2017fast}   
    % & \textbf{BM3D} \cite{makinen2020collaborative} 

    & \textbf{MoDL}\cite{aggarwal2018modl} 
    & \textbf{INFWIDE}\cite{zhang2023infwide} 
    & \textbf{DWDN}\cite{dong2021dwdn} 

    % & \textbf{SVMAP}\cite{dong2021learning} 

    &  \textbf{DECUN} \\ \hline
    SSIM 
    & 0.9162%81844
    & 0.9037%56569

    & 0.8743%23242

    & 0.9096%7508

    & 0.8860%17095

    % & 0.9076%2433

    & 0.8865%56835616799

    & 0.8834%19169

    & 0.9280%96909

    % & 0.7972%16606
    
    & \textbf{0.9294}%63389

        \\ \hline
    PSNR 
    & 30.0522%6346

    & 29.4861%565

    & 28.3476%4232

    & 29.0075%5924

    & 26.7567%5243

    % & 29.8348%6101

    & 27.0891%66706289973

    & 26.9192%9565

    & 31.2046%4257

    % & 26.2477%575

    & \textbf{31.3585}%0182
        \\ \hline
    \end{tabular}
\label{tab7}
\end{table*}

\subsubsection{Computational Comparisons} % Against State of the Art

\begin{table*}[htbp]
\centering
\caption{Per image run time comparison. The best score is in bold fonts.}
\begin{tabular}{|c|c|c|c|c|c|c|c|c|}
    \hline
    & \textbf{EPLL}\cite{zoran2011learning} & \textbf{MLP}\cite{schuler2013machine} & \textbf{IRCNN}\cite{zhang2017} & \textbf{FCNN}\cite{zhang2017learning} & \textbf{FNBD}\cite{son2017fast}   & \textbf{BM3D} \cite{makinen2020collaborative} & \textbf{SVMAP}\cite{dong2021learning} &  \textbf{DECUN} \\ \hline
    Execution time (s) 
    & 68.0663
    & 0.2950
    & 0.4953
    & 0.1485
    & 0.1136
    & 3.6090
    & 2.0086
    & \textbf{0.0614}
        \\ \hline
    \end{tabular}
%\caption{Per image running time comparison. The best score is in bold fonts.}
\label{tab6}
\end{table*}

Table \ref{tab6} summarizes the execution time of each method for processing a typical blurred image of resolution 480 $\times$ 320 and a blur kernel of size 44 $\times$ 44. We include measurements of the running time of each method to deblur this image on GPU.
Specifically, the two benchmark platforms are 1.) Intel Core i7–6900K, 3.20GHz CPU, 62.7GB of RAM, and 2.) an NVIDIA TITAN X GPU. The results are included in Table~\ref{tab6}. Iterative algorithms (EPLL, BM3D, SVMAP) usually run 
hundreds of iterations, which cost much longer time than other methods. In particular, EPLL iterates over each individual patch (8*8 pixels), leading to the longest time to process a single image. On the contrary, deep learning methods (MLP, IRCNN, FCNN, FNBD) process images in an end-to-end fashion with hardware accelerations and are generally much faster. However, deep generic networks typically carry out a large number of filters and operations to ensure adequate modeling power, which induces a higher computational burden and slower inference speed. Among all the methods, DECUN stands out as the fastest-running method and is almost two times as fast as its leading competitor (FNBD). Compared to traditional iterative algorithms, DECUN runs much fewer number of iterations. On the other hand, as the network itself encodes domain knowledge through unrolling, it is free from over-parametrization and much more compact compared to generic off-the-shelf deep networks. As it stands, it outperforms other state-of-the-art techniques while achieving significant computational savings simultaneously.

\begin{figure}[ht]
    \centering
    \begin{subfigure}[t]{0.39\textwidth}
        \centering
        \includegraphics[width=\textwidth,trim=5 5 5 5,clip]{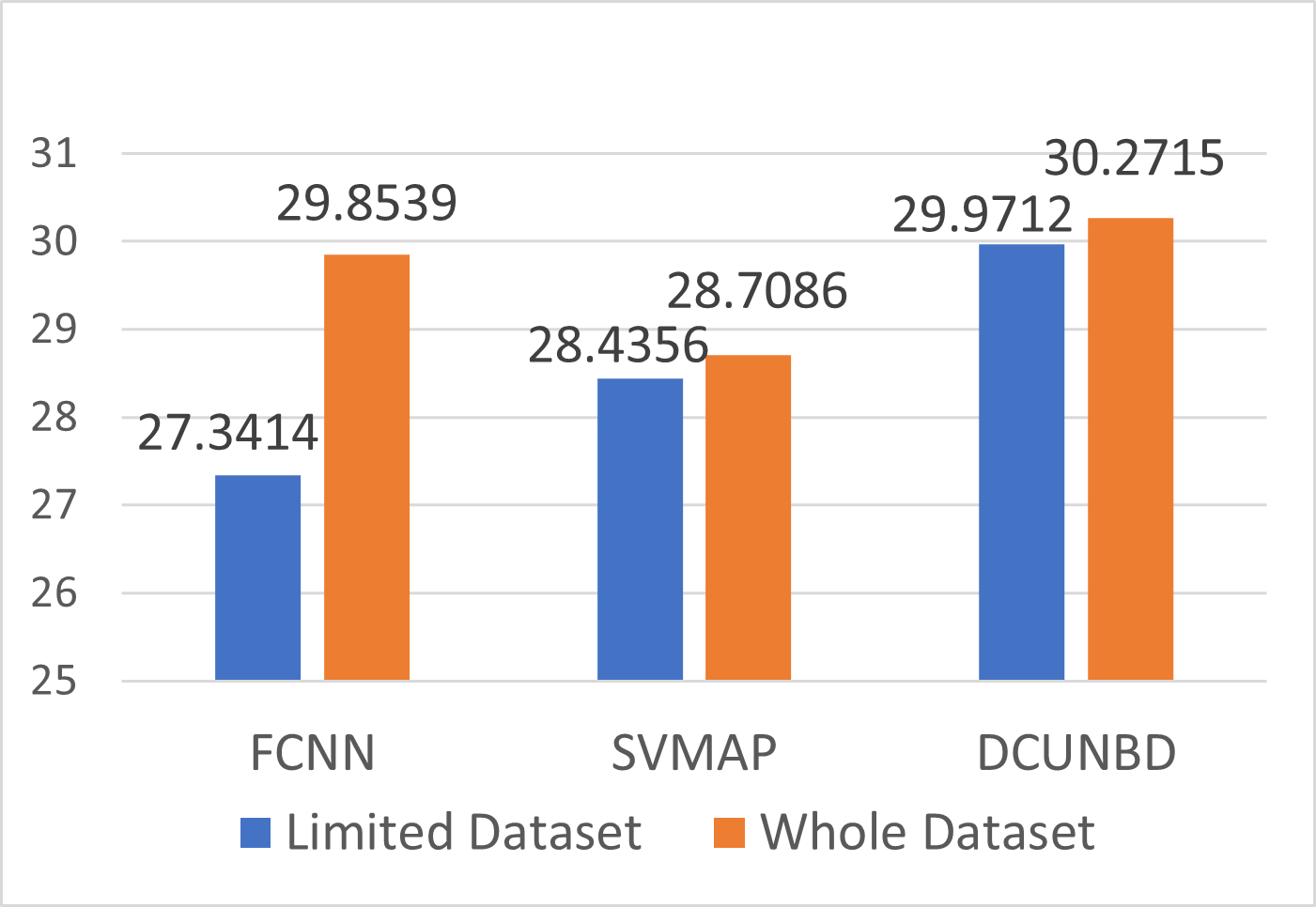}
        \caption{PSNR}
    \end{subfigure}%
    
    \begin{subfigure}[t]{0.39\textwidth}
        \centering
        \includegraphics[width=\textwidth,trim=5 5 5 5,clip]{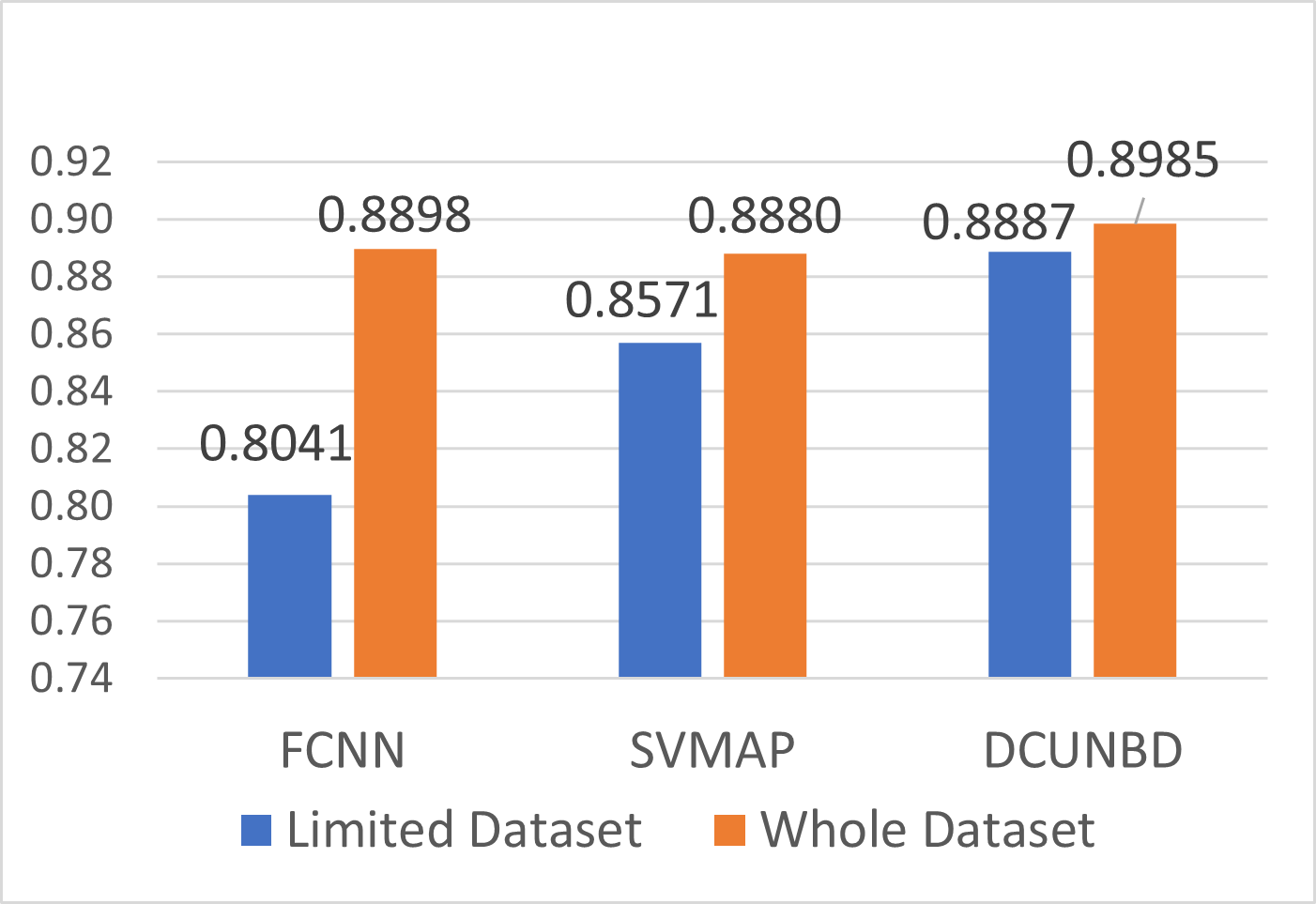}
        \caption{SSIM}
    \end{subfigure}%
    \centering
\caption{Performance evaluation in a limited training set-up.}  % between our method and the state-of-the-art models based on BSD dataset.
\label{fig12}
\end{figure}

\subsubsection{Evaluation in limited training scenarios}
Three methods with the highest score in the previous evaluation (FCNN, SVMAP, and DECUN) are chosen to be trained on a limited dataset with nonlinear kernels. In particular, this limited dataset contains only 10 $\%$ of the training images used for the results in Table \ref{tab5}. The test image set remains the same as the Testing dataset for nonlinear kernels. Fig \ref{fig12} shows the results. The state-of-the-art deep learning approach FCNN with a large number of parameters seems to overfit when training is reduced. Compared to the network trained on the whole dataset, FCNN's performance on the limited dataset has a considerable drop. Because SVMAP utilizes powerful image priors, this technique exhibits expected training robustness and shows a relatively more modest drop. With the fewest parameters among the three methods, DECUN does the best in a limited training scenario. Note that DECUN achieves the smallest gap (esp. in SSIM) between the model trained on the limited dataset and the whole dataset thereby exhibiting superior generalizability. Remarkably, the PSNR/SSIM scores for DECUN in limited training are still higher than competing alternatives, even as they have access to full training.

\section{Conclusion}\label{sec:conclusions}
This paper proposes a DEep, Convergent Unrolling Network (DECUN) based on half-quadratic splitting (HQS) which achieves the simultaneous goals of enhancing performance,  interpretability and  providing theoretical convergence guarantees for the unrolled network. In order to ensure convergence along with substantial network modeling power, we develop a new parametrization scheme and carry out rigorous analytical studies to establish the convergence guarantee of the proposed unrolled network as well as quantify the convergence rate. Our analytical claims are verified through simulation studies that demonstrate convergence. Finally, we experimentally compare DECUN with state-of-the-art deblurring techniques, which demonstrates its favorable performance, particularly in the low-training regime and enhahnced computational benefits. While this study focuses on the non-blind image deblurring problem; a viable future research direction is carrying out similar analysis towards blind image deconvolution problems.
\begin{appendices}
\section{Proof of Theorem~\ref{convergence_l}}\label{sec:proof_thm1}
  \noindent We leverage selected technical results  that help us establish Theorem~\ref{convergence_l}. The following two lemmas are adapted from ~\cite{wang2008new}, which state that both $s^{l}$ and $h^{l}$ are non-expansive operators.
\vspace{-2mm}
\begin{lemma}\label{lemma1}
  For $x\in \mathbb{R}^{2}$, the 2D shrinkage function $s^{l}: \mathbb{R}^{2}\rightarrow \mathbb{R}^{2}$ defined as 
  \begin{equation}
    \label{S_operator}
    \vspace{-1mm}
    s^{l}(\bx)=s_{\frac{1}{\beta^{l}}}(\bx)\triangleq \max\left\{\|\bx\|_2-\frac{1}{\beta^{l}}, 0\right\}\frac{\bx}{\|\bx\|_2},
    \vspace{-1mm}
  \end{equation}
  where $\beta^{l} > 0$, is nonexpansive
  \begin{equation}
    \label{nonexpansiveness}
    \|s^{l}(\bx_{1})-s^{l}(\bx_{2})\|_{2}\leq \|\bx_{1}-\bx_{2}\|_{2}.
  \end{equation}
  Furthermore, if $\|s^{l}(\bx_{1})-s^{l}(\bx_{2})\|_{2}=\|\bx_{1}-\bx_{2}\|_{2}$, then $s^{l}(\bx_{1})-s^{l}(\bx_{2})=\bx_{1}-\bx_{2}$.
  \begin{proof}
   Refer to Proposition 3.1 in~\cite{wang2008new}. For an alternative proof method, the shrinkage function $s^{l}(\cdot)$ is the proximity operator of the $l_1$ norm. Theorem 3 in \cite{cheney1959proximity} shows that the proximity map for a closed convex set is non-expansive.
  \end{proof}
\end{lemma}
\vspace{-4mm}
\begin{lemma}\label{lemma2}
    Assume $\mathcal{N}(\bK) \cap \mathcal{N}(\bar{\bD})=0$, where $\mathcal{N}(\cdot)$ represents the null space of a matrix. Then for any $\bw\neq \widetilde{\bw}\in \mathbb{R}^{n}$, it holds that
  \begin{equation}
    \label{hl_nonexpansiveness}
    \|h^{l}(\bw)-h^{l}(\widetilde{\bw})\|_{2}\leq \|\bw-\widetilde{\bw}\|_{2}.
  \end{equation}
  \begin{proof}
    Refer to Propostion 3.2 in~\cite{wang2008new}.
  \end{proof}
\end{lemma}
\vspace{-4mm}
\begin{lemma}\label{lemma3}
  Let ${\{a_k\}}_k$ be a real non-negative sequence. If $a_{k+1}\leq a_k+\epsilon_k$ for all $k$, $\epsilon_k\geq 0$ and $\sum_{k=1}^\infty\epsilon_k$ converges, then $a_k$ also converges.
\end{lemma}
\begin{proof}
Refer to Lemma 3.1 in~\cite{QFoptimization2001} with the setting of $\chi=1$ and $\beta_{n}$ chosen as a zero sequence with $n \geq 0$. 
\end{proof}

\vspace{-2mm}
We are now ready to prove Theorem 1.
\vspace{-1mm}
\begin{proof}
According to the assumption on $\bD^{l}$, $\beta^{l}$ and the norm bounded sequence ${\{\bE^l\}}_l$, we have $  \lim_{l\rightarrow+\infty} \bD^{l} = \lim_{l\rightarrow+\infty} (\bar{\bD}+\xi_{l}\bE^l)= \bar{\bD} $ and 
$\lim_{l\rightarrow+\infty} \beta^{l}=\lim_{l\rightarrow+\infty} (\bar{\beta}+\gamma_{l})= \bar{\beta}$.
It therefore follows that 
\vspace{-2mm}
\begin{equation}
\begin{split}
\lim_{l\rightarrow+\infty} \bM^{l}&= \lim_{l\rightarrow+\infty} {(\bD^{l})}^{T}\bD^{l}+\frac{\mu}{\beta^{l}}\bK^{T}\bK\\
&= {(\bar{\bD})}^{T}\bar{\bD}+\frac{\mu}{\bar{\beta}}\bK^{T}\bK = \bM^{\ast}.
\vspace{-3mm}
\end{split}
\end{equation}
By continuity of $s^{l}$ and $h^l$, $\forall \!~\bx\in\mathbb{R}^n$, we have $\lim_{l\rightarrow+\infty}s^{l}(\bx) =\max\left\{\|\bx\|_2-\frac{1}{\bar{\beta}}, 0\right\}\frac{\bx}{\|\bx\|_2}\nonumber := s^{\ast}$ and $\lim_{l\rightarrow+\infty} h^{l}(\bx) =\bar{\bD}{(\bM^{\ast})}^{-1}\left[{(\bar{\bD})}^T(\bx)+\frac{\mu}{\bar{\beta}}\bK^T \by\right]\nonumber :=h^{\ast}(\bx)$
% \begin{align}\label{s_*}
%   \lim_{l\rightarrow+\infty}s^{l}(\bx) &=\max\left\{\|\bx\|_2-\frac{1}{\beta^{\ast}}, 0\right\}\frac{\bx}{\|\bx\|_2}\nonumber\\
%                                          &:= s^{\ast}(\bx)
% \end{align}
% and
% \begin{align}\label{h_*}
%   \lim_{l\rightarrow+\infty} h^{l}(\bx) &=\bD^{\ast}{(\bM^{\ast})}^{-1}\left[{(\bD^{\ast})}^T(\bx)+\frac{\mu}{{\beta}^{\ast}}\bK^T \by\right]\nonumber\\
%                                         &:=h^{\ast}(\bx)
% \end{align}
Invoking Lemma~\ref{lemma1} and~\ref{lemma2}, we have $  \|s^l(h^l(\bw^l)) - s^l(h^l(\widetilde{\bw}))\|\leq\|h^l(\bw^l) - h^l(\widetilde{\bw})\|\nonumber \leq\|\bw^l-\widetilde{\bw}\|$.
% \begin{align}
% \end{align}

Furthermore, Lemma~\ref{lemma1} and~\ref{lemma2} ensures non-expansiveness of the shrinkage operator $s^{\ast}$ and the function $h^{\ast}$; suppose $\tilde{\bw}$ is any fixed point of $s^{\ast}\circ h^{\ast}$, that $\tilde{\bw}=s^{\ast}(h^{\ast}(\tilde{\bw}))$, then we have
\begin{equation}\label{proof1_1}
\begin{split}
\vspace{-1mm}
\|\bw^{l+1}-\tilde{\bw}\|&=\|s^{l}(h^{l}(\bw^{l}))-s^{\ast}(h^{\ast}(\tilde{\bw}))\|\\
&=\|s^{l}(h^{l}(\bw^{l}))-s^{l}(h^{l}(\tilde{\bw}))\\
&\qquad\quad+s^{l}(h^{l}(\tilde{\bw}))-s^{\ast}(h^{\ast}(\tilde{\bw}))\|\\
&\leq \|s^{l}(h^{l}(\bw^{l}))-s^{l}(h^{l}(\tilde{\bw}))\|\\
&\qquad\quad+\|s^{l}(h^{l}(\tilde{\bw}))-s^{\ast}(h^{\ast}(\tilde{\bw}))\|\\
& = \|\bw^{l}-\tilde{\bw}\|+\zeta^{l},
\vspace{-2mm}
\end{split}
\end{equation}
where we write $\zeta^{l} = \|s^{l}(h^{l}(\tilde{\bw}))-s^{\ast}(h^{\ast}(\tilde{\bw}))\|$. We are going to show that $\sum_{l=1}^{l}\zeta^{l}$ converges, which indicates that the real non-negative sequence $\|\bw^{l+1}-\tilde{\bw}\|$ converges as well according to Lemma~\ref{lemma3}.
The assertion holds trivially if $\|h^\ast(\widetilde{\bw})\|<\frac{1}{\bar{\beta}}$, as it will be the case that $\|\bD^l\widetilde{\bu}\|<\frac{1}{\beta^l}$ for $l$ sufficiently large. Otherwise, $\|\bar{\bD}\widetilde{\bu}\|>0$. Choose $l$ sufficiently large so that $  |\beta^l-\bar{\beta}|\leq\frac{\bar{\beta}}{2},\qquad \|(\bD^l-\bar{\bD})\widetilde{\bu}\|<\frac{\|\bar{\bD}\widetilde{u}\|}{2},$
% \begin{equation}
% \end{equation}
we have
\begin{equation}
\begin{split}
\label{zeta_l}
\vspace{-1mm}
& \|s^{l}(h^{l}(\tilde{\bw}))-s^{\ast}(h^{\ast}(\tilde{\bw}))\|\\
& =\|h^{l}(\tilde{\bw})-h^{\ast}(\tilde{\bw})-(\chi^{l}(h^{l}(\tilde{\bw}))-\chi^{\ast}(h^{\ast}(\tilde{\bw}))\|\\
& =\Big\|(\bD^{l}-\bar{\bD})\tilde{u}-\left(\frac{1}{\beta^{l}}\frac{\bD^{l}\tilde{u}}{\|\bD^{l}\tilde{u}\|}-\frac{1}{\bar{\beta}}\frac{\bar{\bD}\tilde{u}}{\|\bar{\bD}\tilde{u}\|}\right)\Big\|\\
& \leq \Big\|(\bD^{l}-\bar{\bD})\tilde{u}\Big\|+
\left\|\frac{1}{\beta^{l}}\frac{\bD^{l}\tilde{u}}{\|\bD^{l}\tilde{u}\|}-\frac{1}{\bar{\beta}}\frac{\bar{\bD}\tilde{u}}{\|\bar{\bD}\tilde{u}\|}\right\|\\
&\leq\|\bD^l - \bar{\bD}\|\|\widetilde{\bu}\|+\left\|\left(\frac{1}{\beta^l}-\frac{1}{\bar{\beta}}\right)\frac{\bD^l\widetilde{\bu}}{\|\bD^l\widetilde{\bu}\|}\right\|\\
&+\frac{1}{\bar{\beta}}\left\|\frac{\bD^l\widetilde{\bu}}{\|\bD^l\widetilde{\bu}\|}-\frac{\bar{\bD}\widetilde{\bu}}{\|\bar{\bD}\widetilde{\bu}\|}\right\|\\
&=\|\xi_l\|\|\widetilde{\bu}\|+\frac{|\beta^l-\bar{\beta}|}{\beta^l\bar{\beta}}+\frac{1}{\bar{\beta}}\left\|\frac{\|\bar{\bD}\widetilde{\bu}\|\bD^l\widetilde{\bu}-\|\bD^l\widetilde{\bu}\|\bar{\bD}\widetilde{\bu}}{\|\bD^l\widetilde{\bu}\|\|\bar{\bD}\widetilde{\bu}\|}\right\|\\
&\leq\|\xi_l\|\|\widetilde{\bu}\|+\frac{2|\beta^l-\bar{\beta}|}{{\bar{\beta}}^2}\\
&+\frac{2}{\bar{\beta}}\frac{\|\|\bar{\bD}\widetilde{\bu}\|(\bD^l-\bar{\bD})\widetilde{\bu}+(\|\bar{\bD}\widetilde{\bu}\|-\|\bD^l\widetilde{\bu}\|)\bar{\bD}\widetilde{\bu}\|}{{\|\bar{\bD}\widetilde{\bu}\|}^2}\\
&\leq\|\xi_l\|\|\widetilde{\bu}\|+\frac{2|\beta^l-\bar{\beta}|}{{\bar{\beta}}^2}\\
&+\frac{2}{\bar{\beta}}\frac{\|\bar{\bD}\widetilde{\bu}\|\|(\bD^l-\bar{\bD})\widetilde{\bu}\|+\|\bar{\bD}\widetilde{\bu}-\bD^l\widetilde{\bu}\|\|\bar{\bD}\widetilde{\bu}\|}{{\|\bar{\bD}\widetilde{\bu}\|}^2}\\
&\leq\|\xi_l\|\|\widetilde{\bu}\|+\frac{2|\gamma_l|}{{\bar{\beta}}^2} + \frac{2}{\bar{\beta}}\frac{\|\xi_l\|\|\widetilde{\bu}\|}{\|\bar{\bD}\widetilde{\bu}\|}
\end{split}
\vspace{-4mm}
\end{equation}
where $\tilde{u}=\left({\bar{\bD}}^\mathrm{T} \bar{\bD} + \frac{\mu}{\bar{\beta}}\bK^\mathrm{T} \bK\right)^{-1}\left({\bar{\bD}}^\mathrm{T} \tilde{\bw}+\frac{\mu}{\bar{\beta}}\bK^\mathrm{T} \by\right)$,
%\begin{equation}
%\end{equation}
which is bounded. The convergence of $\sum_l\zeta^l$ follows from that of $\sum_l\|\xi_l\|$ and $\sum_l|\gamma_l|$.

Now that ${\{\bw^l\}}_l$ is bounded, let $\bw^\ast$ be any limit point of it and ${\{\bw^{l_i}\}}_i$ be a subsequence such that $\lim_{i\rightarrow \infty}\bw^{l_i}=\bw^\ast$. Let $l\rightarrow+\infty$,~\eqref{proof1_1} becomes 
\begin{align}\label{proof1_3}
\vspace{-1mm}
&\lim_{l\rightarrow+\infty}\|\bw^{l+1}-\tilde{\bw}\|\\
&\leq \lim_{l\rightarrow+\infty}\|\bw^{l}-\tilde{\bw}\|=\lim_{i\rightarrow\infty}\|\bw^{l_{i}}-\tilde{\bw}\|=\|\bw^{*}-\tilde{\bw}\|
\vspace{-1mm}
\end{align}
which implies that all limit point of $\{\bw^{k}\}$, if more than one, have an equal distance to $\tilde{\bw}$. On the other hand,
\begin{align*}
\vspace{-1mm}
  &\|\bw^{l_i+1}-s^\ast(h^\ast(\bw^\ast))\|\\
  &=\|s^{l_i}(h^{l_i}(\bw^{l_i}))-s^\ast(h^\ast(\bw^\ast))\|\\
  &\leq\|s^{l_i}(h^{l_i}(\bw^{l_i}))-s^{l_i}(h^{l_i}(\bw^\ast))\|\\
  &\quad+\|s^{l_i}(h^{l_i}(\bw^\ast))-s^\ast(h^\ast(\bw^\ast))\|\\
  &\leq\|\bw^{l_i}-\bw^\ast\|+\|s^{l_i}(h^{l_i}(\bw^\ast))-s^\ast(h^\ast(\bw^\ast))\|,
  \vspace{-1mm}
\end{align*}
which implies $\lim_{i\rightarrow\infty}\bw^{l_{i}+1}=s^{\ast}(h^{\ast}(\bw^{\ast}))$,
so that $s^{\ast}(h^{\ast}(\bw^{\ast}))$ is also a limit point of $\{\bw^{k}\}$ which is required to have the same distance to $\tilde{\bw}$ as $\bw^{\ast}$ does, that is $\|\bw^{\ast}-\tilde{\bw}\|=\|s^{\ast}(h^{\ast}(\bw^{\ast}))-\tilde{\bw}\|=\|s^{\ast}(h^{\ast}(\bw^{\ast}))-s^{\ast}(h^{\ast}(\tilde{\bw}))\|$.
Since $\tilde{\bw}$ is any fixed point of $s^{\ast}(h^{\ast}(\cdot))$, replacing $\tilde{\bw}$ with $\bw^{\ast}$ in~\eqref{proof1_3} gives rise to $  \lim_{l\rightarrow \infty}\|\bw^{l}-\bw^{*}\|=\lim_{i\rightarrow\infty}\|\bw^{l_{i}}-\bw^{*}\|=\|\bw^{*}-\bw^{*}\|=0$,
% \begin{equation}
% \end{equation}
which implies
\begin{equation}\label{w_convergence}
  \lim_{l\rightarrow \infty} \bw^{l} =\bw^{*}.
\end{equation}
Combining~\eqref{eq:u_update} and~\eqref{w_convergence} leads to
\begin{align*}
&\lim_{l\rightarrow+\infty} \bu^{l}\\
&~\!=\lim_{l\rightarrow+\infty}{\left({\bD^{l}}^\mathrm{T} \bD^{l} + \frac{\mu}{\beta^{l}}\bK^\mathrm{T} \bK\right)}^{-1}
           \!\!\left({\bD^{l}}^\mathrm{T} \bw^{l}+\frac{\mu}{{\beta}^{l}}\bK^\mathrm{T} \by\right)=\bu^{\ast}
\end{align*}
which completes the proof of Theorem~\ref{convergence_l}.
\end{proof}

\section{Proof of Theorem.~\ref{thm:convergence_rate}}\label{sec:proof_thm2}
According to~\eqref{eq:function_h_l},~\eqref{w_l1} as well as the expansiveness of $s^{l}$ in \eqref{nonexpansiveness}, we have
\begin{equation}
\label{theorm3_1}
\vspace{-2mm}
\begin{split}
&\|\mathbf{w}^{l+1}-\mathbf{w}^{*}\|\\
&\quad= \|s^{l}(h^{l}(\mathbf{w}^{l}))-s^{\ast}(h^{\ast}(\mathbf{w}^{\ast}))\|\\
&\quad=\|s^{l}(h^{l}(\mathbf{w}^{l}))-s^{l}(h^{l}(\mathbf{w}^{\ast}))+s^{l}(h^{l}(\mathbf{w}^{\ast}))-s^{\ast}(h^{\ast}(\mathbf{w}^{\ast}))\|\\
&\quad\leq \|h^{l}(\mathbf{w}^{l})-h^{l}(\mathbf{w}^{*})\|
+\|s^{l}(h^{l}(\mathbf{w}^{\ast}))-s^{\ast}(h^{\ast}(\mathbf{w}^{\ast}))\|.
\end{split}
\vspace{-2mm}
\end{equation}
For the term $ \|h^{l}(\mathbf{w}^{l})-h^{l}(\mathbf{w}^{*})\|$ alongside with \eqref{eq:function_h_l}, we have
\begin{equation}
    \begin{split}
    \vspace{-3mm}
     \|h^{l}(\mathbf{w}^{l})&-h^{l}(\mathbf{w}^{*})\|^2\\
     &=\|\bD^{l}(\bM^{l})^{-1}(\bD^{l})^\mathrm{T}(\mathbf{w}^{l}-\mathbf{w}^{*})\|^2\\
&= \|\mathbf{G}^{l}(\mathbf{w}^{l}-\mathbf{w}^{*})\|^2\\
&= (\mathbf{w}^{l}-\mathbf{w}^{\ast})^{\mathrm{T}}(\mathbf{G}^{l})^{2}(\mathbf{w}^{l}-\mathbf{w}^{\ast})\\
& \leq \rho((\mathbf{G}^{l})^{2})\|\mathbf{w}^{l}-\mathbf{w}^{\ast}\|^{2}
    \end{split}
    \vspace{-3mm}
\end{equation}
so that
\begin{equation}
\begin{split}
\label{first_term}
\vspace{-1mm}
\|h^{l}(\mathbf{w}^{l})-h^{l}(\mathbf{w}^{*})\|\leq \sqrt{\rho((\mathbf{G}^{l})^{2})}\|\mathbf{w}^{l}-\mathbf{w}^{\ast}\|
\vspace{-1mm}
\end{split}
\end{equation}
where $\mathbf{G}^{l}=\bD^{l}(\bM^{l})^{-1}(\bD^{l})^\mathrm{T}$, and $\rho((\mathbf{G}^{l})^{2})$ is the spectral radii of the matrix $(\mathbf{G}^{l})^{2}$ satisfying $\rho((\mathbf{G}^{l})^{2})\leq\rho(\mathbf{G}^{l})\leq 1$. Then we are going to calculate the upper bound of the term $\|s^{l}(h^{l}(\mathbf{w}^{\ast}))-s^{\ast}(h^{\ast}(\mathbf{w}^{\ast}))\|$. 
%Since $h^{l}(\mathbf{w}^{\ast})=\{h^{l}_{I}(\mathbf{w}^{\ast}),h^{l}_{J}(\mathbf{w}^{\ast})\}$ in which $h^{l}_{I}(\mathbf{w}^{\ast})$ satisfying $\|s^{l}(h^{l}_{I}(\mathbf{w}^{\ast}))-s^{\ast}(h^{\ast}_{I}(\mathbf{w}^{\ast}))\|=0$ according to the definition of $I$, so that 
%\begin{equation}
%    \begin{split}
%    \label{theorem3_step2}
%\|s^{l}(h^{l}(\mathbf{w}^{\ast}))&-s^{\ast}(h^{\ast}(\mathbf{w}^{\ast}))\|\\
%& = \|s^{l}(h^{l}_{J}(\mathbf{w}^{\ast}))-s^{\ast}(h^{\ast}_{J}(\mathbf{w}^{\ast}))\|
%    \end{split}
%\end{equation}
%where $\|h^{l}_{J}(\mathbf{w}^{\ast})\|\geq \frac{1}{\beta^{l}}$ and $\|h^{\ast}_{J}(\mathbf{w}^{\ast})\|\geq \frac{1}{\beta^{\ast}}$. 
According to the definition of $s^{l}$ in \eqref{eq:w_update} and \eqref{zeta_l} by replacing $\tilde{u}$ with $\bu^{\ast}$, we have
\begin{equation}
    \begin{split}
    \label{second_term}
    \vspace{-1mm}
\|s^{l}(h^{l}&(\mathbf{w}^{\ast}))-s^{\ast}(h^{\ast}(\mathbf{w}^{\ast}))\|\\
%&\qquad=\|s^{l}(h^{l}_{J}(\mathbf{w}^{\ast}))-s^{\ast}(h^{\ast}_{J}(\mathbf{w}^{\ast}))\|\\
 &\qquad\leq\|\xi_l\|\|\bu^{\ast}\|+\frac{2|\gamma_l|}{{(\bar{\beta})}^2} + \frac{2}{\bar{\beta}}\frac{\|\xi_l\|\|\bu^{\ast}\|}{\|\bar{\bD}\bu^{\ast}\|}.
 \vspace{-1mm}
    \end{split}
\end{equation}
Combining the condition of $\|h^{\ast}_{J}(\mathbf{w}^{\ast})\|=\|\bar{\bD}\bu^{\ast}\|\geq \frac{1}{\bar{\beta}}$, \eqref{second_term} further equals to
\begin{equation}
    \begin{split}
    \label{second_term_2}
    \vspace{-3mm}
    \|s^{l}(h^{l}(\mathbf{w}^{\ast}))-s^{\ast}(h^{\ast}(\mathbf{w}^{\ast}))\|\leq 3\|\xi_l\|\|\bu^{\ast}\|+\frac{2|\gamma_l|}{{(\bar{\beta})}^2}
    \end{split}
     \vspace{-3mm}
\end{equation}
where  $\bu^{\ast}=\left({(\bar{\bD})^\mathrm{T} \bar{\bD}} + \frac{\mu}{\bar{\beta}}\bK^\mathrm{T} \bK\right)^{-1}\left((\bar{\bD})^\mathrm{T}\bw^{\ast}+\frac{\mu}{\bar{\beta}}\bK^\mathrm{T} \by\right).$
% \begin{equation}
% \end{equation}
%Since we are interested in the asymptotic behavior of \eqref{eq:u_update} and \eqref{eq:w_update}, without loss of generality, we afterwards assume that
%$\mathbf{w}^{l}_{I}\triangleq \mathbf{w}^{*}_{I}=\mathbf{0}$, and 
Combining \eqref{first_term} and \eqref{second_term_2}, 
%for $\mathbf{w}^{l}=\{\mathbf{w}^{l}_{J},\mathbf{w}^{l}_{I}\}$, 
\eqref{theorm3_1} becomes
\begin{equation}
\begin{split}
 \vspace{-1mm}
\label{wl+1}
\|\mathbf{w}^{l+1}-\mathbf{w}^{*}\|\leq \lambda_{l}\|\bw^{l}-\bw^{*}\|+ b_{l}
\end{split}
 \vspace{-1mm}
\end{equation}
where $\lambda_{l}=\sqrt{\rho((\mathbf{G}^{l})^{2}))}$, $b_{l} = 3\|\xi_l\|\|\bu^{\ast}\|+\frac{2|\gamma_l|}{{(\bar{\beta})}^2}$
and $\bu^{\ast}$ is defined in \eqref{u_J}.

We next show that, $\exists\lambda_{\max}$ so that $\lambda_{l}\leq\lambda_{\max}<1,\forall l$. Indeed, since $\mathbf{G}^l$ is symmetric positive semi-definite, $\rho(\mathbf{G}^l)=\left\|\mathbf{G}^l\right\|$ is continous over $\mathbf{D}^l$ and $\beta^l$, and hence $\lim_{l\rightarrow\infty}\mathbf{G}^l=\mathbf{G}^\ast=\bar{\bD}{\left({\bar{\bD}}^T\bar{\bD}+\frac{\mu}{\bar{\beta}}\bK^T\bK\right)}^{-1}\bar{\bD}$. Choose $l_0\in\mathbb{N}$ sufficient large so that
$
  \left|\rho(\mathbf{G}^l)-\rho(\mathbf{G}^\ast)\right|<\frac{1}{2}\left[1 - \rho(\mathbf{G}^\ast)\right],\quad\forall l\geq l_0,
$
then $\rho(\mathbf{G}^l)<\frac{1}{2}\left[1 + \rho({\mathbf{G}}^\ast)\right] < 1,\quad\forall l\geq l_0$. Pick $\lambda_{\max}=\max\left\{\lambda_1,\lambda_2,\dots,\lambda_{l_0},\sqrt{\frac{1}{2}\left[1 + \rho({\mathbf{G}}^\ast)\right]}\right\}$, then we have $\lambda_l\leq\lambda_{\max}<1,\,\forall l$. Iterating~\eqref{wl+1} with $l$ times, we have
\begin{equation}
\begin{split}
\label{convergence_rate_2}
\vspace{-3mm}
    &\|\mathbf{w}^{l+1}-\mathbf{w}^{*}\|\\
    &\quad\leq (\lambda_{\max})^{l+1}\|\mathbf{w}^{0}-\mathbf{w}^{*}\|+\sum_{i=0}^{l}(\lambda_{\max})^{l-i}b_{i}\\
    &\quad=(\lambda_{\max})^{l+1}\|\mathbf{w}^{0}-\mathbf{w}^{*}\|+ (\lambda_{\max})^{l}\otimes b_{l}\\
    &\quad=(\lambda_{\max})^{l+1}\|\mathbf{w}^{0}-\mathbf{w}^{*}\|+\mathcal{F}^{-1}\left(\widehat{(\lambda_{\max})^{l}}\odot \widehat{b_{l}}\right)\\
    &\quad=(\lambda_{\max})^{l+1}\|\mathbf{w}^{0}-\mathbf{w}^{*}\|+\mathcal{F}^{-1}\left(\frac{B(w)}{1-\lambda_{\max}e^{-jw}}\right)\\   
\end{split}
 \vspace{-4mm}
\end{equation}
where $\otimes$ is the discrete convolution operator, $\widehat{(\lambda_{\max})^{l}}=\frac{1}{1-\lambda_{\max}e^{-jw}}$, and $B(w)$ is the Discrete-time Fourier transform of ${(\lambda_{\max})}^{l}$ and $b_{l}$ respectively.
\end{appendices}
\bibliographystyle{IEEEbib}
\bibliography{TCI}

\end{document}